%% file: CSMA+Energy_report.tex
\newif\ifreport\reporttrue
\documentclass[journal,10pt,twocolumns]{IEEEtran}
\usepackage{algorithmic}
\usepackage{amsmath}
\usepackage{amsfonts}
\usepackage{amssymb}

\usepackage{graphicx,enumerate}
\usepackage[center]{caption}
\usepackage{epsfig,latexsym,amsmath,amsfonts}
\usepackage{mathtools}

\usepackage{amssymb}
\usepackage{placeins}
\usepackage{subfigure}
\usepackage{verbatim}
\usepackage{cite,url}
\usepackage{epsf,bm}
\usepackage{color}
\usepackage{epstopdf}
\usepackage{romannum,hyperref}
\usepackage[onehalfspacing]{setspace}
\usepackage{dsfont}
\usepackage{bbm}
\usepackage{stfloats}
\usepackage[utf8]{inputenc}
\usepackage[english]{babel}
\usepackage{amsthm}
\newtheorem{definition}{Definition}
\newtheorem{remark}{Remark}

\usepackage[utf8]{inputenc}
\usepackage[english]{babel}

\newtheorem{theorem}{Theorem}
\newtheorem{assumption}{Assumption}
\newtheorem{lemma}[theorem]{Lemma}

\newtheorem{corollary}[theorem]{Corollary}

\newcommand{\br}{\mathbf{r}}
\newcommand{\cX}{\mathcal{X}}
\newcommand{\te}{\theta}

\newcommand{\bN}{\mathbb{N}}
\newcommand{\bR}{\mathbb{R}}
\newcommand{\bP}{\mathbb{P}}
\newcommand{\bE}{\mathbb{E}}
\newcommand{\cC}{\mathcal{C}}

\newcommand{\bx}{\mathbf{x}}
\newcommand{\by}{\mathbf{y}}
\newcommand{\cF}{\mathcal{F}}

\newcommand{\cK}{\mathcal{K}}
\newcommand{\cE}{\mathcal{E}}
\newcommand{\cG}{\mathcal{G}}

\newcommand{\cS}{\mathcal{S}}

\newcommand{\ust}{^{\star}}

\newcommand{\bs}{\mathbf{s}}
\newcommand{\id}{\mathbbm{1}}
\newcommand{\cB}{\mathcal{B}}

\usepackage[lined, boxed, linesnumbered, ruled]{algorithm2e}
\IEEEoverridecommandlockouts
\begin{document}
\pagenumbering{arabic}
\title{Low-Power Status Updates via Sleep-Wake Scheduling}  



\author{\large Ahmed M. Bedewy, Yin Sun, \emph{Senior Member, IEEE}, Rahul Singh, and Ness B. Shroff, \emph{Fellow, IEEE}
\thanks{This paper was presented in part at ACM MobiHoc 2020 \cite{bedewy2020optimizing}.}
\thanks {This work has been supported in part by ONR grants N00014-17-1-2417 and N00014-15-1-2166, Army Research Office grants W911NF-14-1-0368 and MURI W911NF-12-1-0385,  National Science Foundation grants CNS-1446582, CNS-1421576, CNS-1518829, and CCF-1813050, and a grant from the Defense Thrust Reduction Agency HDTRA1-14-1-0058.}
\thanks{A. M. Bedewy is with the  Department  of  ECE,  The  Ohio  State  University, Columbus, OH 43210 USA (e-mail:  bedewy.2@osu.edu).}
\thanks{Y.  Sun  is  with  the  Department  of  ECE,  Auburn  University,  Auburn,  AL 36849 USA (e-mail:  yzs0078@auburn.edu).}
\thanks{R.  Singh  is  with  the Department of ECE, Indian Institute of Science, Bangalore 560012, India (e-mail:  rahulsingh@iisc.ac.in).}
\thanks{N. B.  Shroff  is  with  the  Department  of  ECE and  the  Department  of  CSE, The Ohio State University,  Columbus, OH 43210 USA  (e-mail:  shroff.11@osu.edu).}
}

\maketitle
\input{sections/abstract}
\newpage
\input{sections/intro}

\input{sections/sysmodel}
\input{sections/main_result}

\input{sections/proof_main_result}
\input{sections/learning_algo}
\input{sections/Numerical_results}

\input{sections/conclusion}
\bibliographystyle{IEEEbib}
\bibliography{MyLib}
\input{sections/appendices_sec}

\input{sections/learning_proof}
\end{document}

%% file: sections/abstract.tex
\begin{abstract}

We consider the problem of optimizing the freshness of status updates that are sent from a large number of  low-power sources to a common access point. The source nodes utilize carrier sensing to reduce collisions and adopt an asynchronized sleep-wake scheduling strategy to achieve a target network lifetime (e.g., 10 years). We use \emph{age of information} (AoI) to measure the freshness of status updates, and design sleep-wake parameters for minimizing the weighted-sum peak AoI of the sources, subject to per-source battery lifetime constraints. When the sensing time  (i.e., the time duration of carrier sensing) is zero, this sleep-wake design problem can be solved by resorting to a two-layer nested convex optimization procedure; however, for positive sensing times, the problem is non-convex. We devise a low-complexity solution to solve this problem and prove that, for practical sensing times that are short, the solution is within a small gap from the optimum AoI performance. When the mean transmission time of status-update packets is unknown, we devise a reinforcement learning algorithm that  adaptively performs the following two tasks in an ``efficient way'': a) it learns the unknown parameter, b) it also generates efficient controls that make channel access decisions. We analyze its performance by quantifying its ``regret'', i.e., the sub-optimality gap between its average performance and the average performance of a controller that knows the mean transmission time.  Our numerical and NS-3 simulation results show that our solution can indeed elongate the batteries lifetime of information sources, while providing a competitive AoI performance.


\end{abstract}

%% file: sections/intro.tex
\section{Introduction}\label{Int}
In applications such as networked monitoring and control systems, wireless sensor networks, autonomous vehicles, it is crucial for the destination node to receive timely status updates so that it can make accurate decisions.  \emph{Age of information} (AoI) has been used to measure the freshness of status updates. 
More specifically, AoI \cite{KaulYatesGruteser-Infocom2012} is the age of the freshest update at the destination, i.e., it is the time elapsed since the freshest received update was generated. It should be noted that optimizing traditional network performance metrics, such as throughput or delay, do not attain the goal of timely updating. For instance, it is well known that AoI could become very large when the offered load is high or low \cite{KaulYatesGruteser-Infocom2012}.  In other words, AoI captures the information lag at the destination, and is hence more apt for achieving the goal of timely updates. Thus, AoI has recently attracted a lot of interests (see \cite{Yin_book,yates2020age}  and references therein).

In a variety of information update systems, energy consumption is also a critical concern. For example, wireless sensor networks are used for monitoring crucial natural and human-related activities, e.g. forest fires, earthquakes, tsunamis, etc. Since such applications often require the deployment of sensor nodes in remote or hard-to-reach areas, they need to be able to operate unattended for long durations. Likewise, in medical sensor networks,  battery replacement/recharging involves a series of medical procedures, leading to disutility to patients. Hence,  energy consumption must be constrained in order to support a long battery life of 10-15 years~\cite{timmons2004analysis}\ifreport\footnote{The computations performed in \cite{timmons2004analysis} are based on the specifications of commercially used devices. For example, the used transceiver is 2.4 GHz chipset from Chipcon, the CC2420 \cite{chipson}, and the used microcontroller is the Motorola 8-bit microcontroller MC9508RE8 \cite{motorola}. For more detail about the supply voltage and current consumption, please see the aforementioned references.}\fi. 
For networks serving such real-time applications, prolonging battery-life is  crucial. Existing works on multi-source networks, e.g., \cite{yates2017status,talak2018distributed,li2013throughput,kadota2016minimizing_journal,hsu2017scheduling_2,jiang2018timely,kadota2018optimizing,talak2018optimizing2,aphermedis_he2017optimal,DBLP:journals/ton/GuoSKN18,  Yin_multiple_flows,kadota2016minimizing,kadota2016minimizing_journal,singh2015index}, focused exclusively on minimizing the AoI and overlooked the need to reduce power consumption. This motivates us to derive scheduling algorithms that achieve a trade-off between the competing tasks of minimizing AoI and reducing the energy consumption in multi-source networks.

Additionally, some status-update systems consist of  a large number (e.g., hundreds of thousands) of densely packed wireless nodes, which are  serviced by  a single access point (AP). Examples include massive machine-type communications~\cite{kowshik2019energy}. The dataloads in such ``dense networks'' \cite{kowshik2019energy,kowshik2019fundamental} are created by applications such as home security and automation, oilfield and pipeline monitoring, smart agriculture, animal  and livestock tracking, etc. This introduces high variability in the data packet sizes so that the transmission times of data packets are random. Thus, scheduling algorithms  designed for time-slotted systems with a fixed transmission duration, are not applicable to these systems. Besides that, synchronized scheduler for time-slotted systems are feasible when there are relatively few sources and each source has sufficient energy. However, if there are a huge number of sources, the signaling overhead could be quite high. Since,  each source may have limited energy and low traffic rate,  the system could be highly inefficient.
This motivates us to design asynchronized medium access protocols that coordinate the transmissions of multiple conflicting transmitters connected to a single AP.


Towards that end, we consider a wireless network with $M$ sources that contend for channel access and communicate their update packets to an AP. Each source is equipped with a battery that may get charged by a renewable source of energy, e.g., solar. Moreover, each source employs a sleep-wake scheduling scheme \cite{chen2013life} under which the source transmits a packet if the channel is  idle; and sleeps if either: (i) The channel is busy, (ii) it has completed a packet transmission. This enables each source to save the precious battery energy by switching off  when it is unlikely to gain channel access for packet transmissions.  However, since a source cannot transmit during the sleep period, this causes the AoI to increase. We  carefully design the sleep-wake parameters to minimize the  weighted-sum peak age of  the sources, while ensuring that the battery lifetime constraint of each source is satisfied. 

\subsection{Related Works}
There have been significant recent efforts on analyzing the AoI performance of popular queueing service disciplines, e.g., the First-Come, First-Served (FCFS) \cite{KaulYatesGruteser-Infocom2012}
 Last-Come, First-Served (LCFS) with and without preemption \cite{RYatesTIT16_2},
and queueing systems with packet management \cite{CostaCodreanuEphremides_TIT}.
 In \cite{age_optimality_multi_server,Bedewy_NBU_journal_2,multihop_optimal,bedewy2017age_multihop_journal_2,Yin_multiple_flows}, the age-optimality of Last-Generated, First-Served (LGFS)-type policies in multi-server and multi-hop networks was established, where it was shown that these policies can minimize any non-decreasing functional of the age processes. 
The design of data sampling and transmission in information update systems was investigated in \cite{SunJournal2016,sun2018sampling_2}, where sampling policies were derived to minimize nonlinear age functions in single source systems. In \cite{sun2018sampling_2}, it was shown that  a variety of information freshness metrics can be represented as monotonic functions of the age.  The studies in \cite{SunJournal2016,sun2018sampling_2} were later extended to a multi-source scenario in \cite{multi_source_bedewy_2,Bedewy_multisource_journal_1}. 

Designing scheduling policies for minimizing AoI in multi-source networks has recently received increasing attention, e.g., \cite{DBLP:journals/ton/GuoSKN18,yates2017status,talak2018distributed,li2013throughput,kadota2016minimizing_journal,hsu2017scheduling_2,jiang2018timely,kadota2018optimizing,talak2018optimizing2,aphermedis_he2017optimal}. Of particular interest are those pertaining to designing distributed scheduling policies \cite{yates2017status,talak2018distributed,li2013throughput,hsu2017scheduling_2,kadota2016minimizing_journal,jiang2018timely}. The work in \cite{yates2017status} considered a slotted ALOHA-like random access scheme in which each node accesses the channel with a certain access probability. These probabilities were then optimized in order to minimize the AoI. However, the model of~\cite{yates2017status} allows multiple interfering users to gain channel access simultaneously, and hence allows for the collision. The authors in \cite{talak2018distributed} generalized the work in~\cite{yates2017status} to a wireless network in which the interference is described by a general interference model. 
The Round Robin or Maximum Age First policy was shown to be (near) age-optimal for different system models, e.g., in \cite{li2013throughput,hsu2017scheduling_2,kadota2016minimizing_journal,jiang2018timely,Yin_multiple_flows}.

Carrier sensing distributed medium access mechanisms, e.g., Carrier Sense Multiple Access (CSMA), have been widely adopted in many wireless networks; see \cite{yun2012optimal,rahul_prk}  for a recent survey. There has been an interest in designing CSMA-based scheduling schemes that optimize the AoI \cite{maatouk2019minimizing,wang2019broadcast}. In \cite{maatouk2019minimizing}, the authors designed an idealized CSMA (similar to that in \cite{jiang2010distributed}) to minimize the AoI with an exponentially distributed packet transmission times.  In \cite{wang2019broadcast}, the authors designed a slotted Carrier Sense Multiple Access/Collision-Avoidance (CSMA/CA) (similar to that in \cite{bianchi2000performance}) to minimize the broadcast age of information, which is defined, from a sender's perspective, as the age of the freshest successfully broadcasted packet. Contrary to these works, the sleep-wake scheduling scheme proposed by us emphasizes on reducing the cumulative energy consumption in multi-source networks in addition to minimizing the total weighted AoI. Moreover, in our study, transmission times are not necessarily random variables with some commonly used parametric density \cite{maatouk2019minimizing}, or deterministic \cite{wang2019broadcast}, but can be any generally distributed random variables with finite mean.


\subsection{Key Contributions}
Our key contributions are summarized as follows:
\begin{itemize}

\item In our model, sources utilize an asynchronized sleep-wake scheduling strategy to achieve an extended battery lifetime. We aim at designing the mean sleeping period of each source,  which controls its channel access probability, in order to minimize the total weighted average peak age of the sources while simultaneously meeting per-source battery lifetime constraints. Although, the aforementioned optimization problem is non-convex,  we devise a solution. In the regime for which the sensing time is negligible compared to the packet transmission time, the proposed solution is near-optimal (Theorem \ref{thm_case_sum_bi_ge_1} and Theorem \ref{thm_case_sum_bi_le_1}). Our near-optimality results hold for general distributions of the packet transmission times.

\item  We propose an algorithm that can be easily implemented in many practical control systems. In particular, our solution  requires the knowledge of only two variables in its implementation. These two variables are functions of the network parameters. An implementation procedure to compute these two variables  is provided. 


\item As the ratio between the sensing time and the packet transmission time reduces to zero, we show that the age performance of our proposed algorithm is as good as that of the optimal synchronized scheduler (e.g., for time-slotted systems). 


\item Finally, since our solution is a function of the mean transmission time of data packets, the network operator needs to know this quantity in order to implement the algorithm. The transmission times however depend upon the environmental conditions, which in turn are hard to predict before the system operation begins. To overcome this challenge, we develop a reinforcement learning (RL)~\cite{sutton1998reinforcement,jaksch2010near,singh2020learning} algorithm that maintains an estimate of the (unknown) mean transmission time, and then utilizes this estimate in order to derive a solution that is ``seemingly optimal'' for the true system. We show that the regret of the proposed RL algorithm scales as $\tilde{O}(\sqrt{H})$,\footnote{$\tilde{O}$ hides factors that are logarithmic in $H$.}  where $H$ is the operating time horizon.



\end{itemize}

%% file: sections/sysmodel.tex
\section{Model and Formulation}\label{sysmod}
\begin{figure*}[t]
\includegraphics[scale=.8]{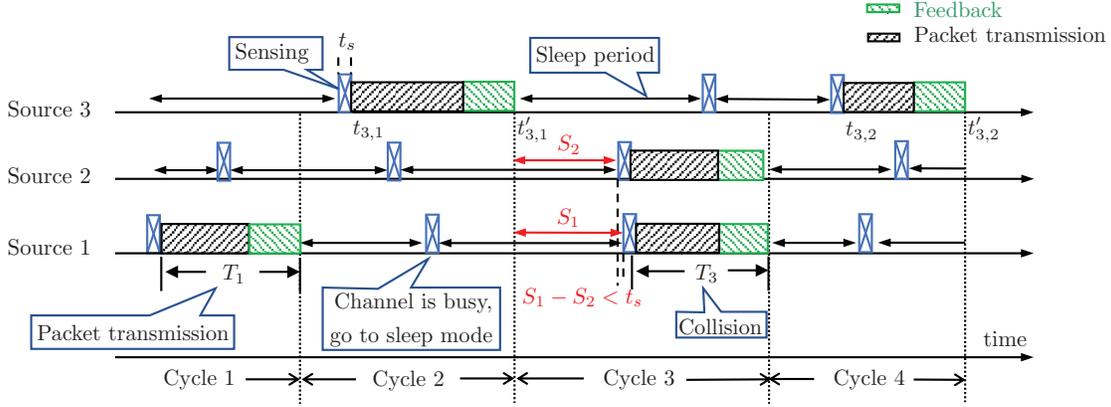}
\centering
\captionsetup{justification=justified}
\caption{Illustration of the sleep-wake cycles. In Cycles 1-2, we have successful packet transmissions. Let $S_1$ and $S_2$ represent the remaining sleeping times of Sources 1 and 2,  respectively, after a successful transmission. Then, a collision occurs in Cycle 3 because 
 the difference between wake-up times of Sources 1 and 2 is less than $t_s$, i.e., $S_1-S_2<t_s$. As we can observe, each cycle consists of an idle period before a transmission/collision event.}
 \label{sleep_wake_cycle}
\end{figure*}
\subsection{Network Model and Sleep-wake Scheduling}
Consider a wireless network composed of $M$ source nodes, each observing a time-varying signal. The sources generate update packets of the observed signals  and send the packets to an access point (AP) over a shared spectrum band. If multiple sources transmit packets simultaneously, a packet collision occurs and these packet transmissions fail. 


The sources use a sleep-wake scheduling scheme to access the shared spectrum, where each source switches between a sleep mode and a transmission mode over time, according the following rules: Upon waking from the sleep mode, a source first performs carrier sensing to check whether the channel is occupied by another source, as illustrated in Figure \ref{sleep_wake_cycle}. The time duration of carrier sensing is denoted as $t_s$, which is sufficiently long to ensure a high sensing accuracy.
If the channel is sensed to be busy, the source enters the sleep mode directly; otherwise,  the source generates an update packet and sends it over the channel. The  source hereafter goes back to the sleep mode. 

In the above sleep-wake scheduling scheme, if two sources start transmitting within a time duration of $t_s$, then their sensing periods are overlapping and they may not be able to detect the transmission of each other. In order to obtain a robust system design, we consider that they cannot detect each other's transmission and a collision occurs. Upon completing a packet transmission,  sources switch to the reception mode and wait for an acknowledgement  (ACK) that indicates the outcome of their transmissions (successful transmission or collision). They then go back to the sleep mode.

A \emph{sleep-wake cycle}, or simply a \emph{cycle}, is defined as the time period between the ends of two successive packet transmission or collision events. Each cycle consists of an idle period and a transmission/collision period\footnote{To make the sleep-wake scheduling problem solvable analytically, we make several approximations. For example, in 802.11b frame structure, there exists a Short Inter-frame Space (SIFS) between the packet transmission frame and the ACK frame (i.e., the CTS frame). If another source wakes up during the SIFS, then it may not detect the transmission/ACK frames, leading to unexpected collisions. In our analytical model, such collision events are omitted. In other words, we suppose that each cycle must start with an idle period, where all sources are in the sleep mode, followed by a transmission/collision period.  NS-3 simulation results will be provided in Section \ref{ns3sim} to show that these approximations have a negligible impact on the age performance of our solution.}. As depicted in Figure \ref{sleep_wake_cycle}, the packet transmissions in Cycles 1-2 are successful, but a collision occurs in Cycle 3 because Sources 1 and 2 wake up within a short duration $t_s$.   

 We use $T_j,j\in \{1,2,\ldots\}$ to represent the time incurred by the $j$-th packet transmission or collision event, which includes transmission/collision time and feedback delays. For example, in Figure~\ref{sleep_wake_cycle}, $T_1$ is the time duration of the packet transmission event by Source~1, while $T_3$ is the time duration of the collision event between Source~1 and 2. We assume that the $T_j$'s are i.i.d. for all transmission and collision events, with a general distribution. This assumption does not hold in practice. Nonetheless, NS-3 simulation results  in Section \ref{ns3sim} show that this assumption has a negligible impact on the performance of the proposed algorithm.  When there is no confusion, we omit the subscript $j$ of $T_j$ for simplicity, and use $T$ to denote the transmission/collision time, which is assumed to have a finite mean, i.e., $E[T]<\infty$. The sleep periods of source~$l$ are exponentially distributed random variables with mean value $\mathbb{E}[T]/r_l$ and are independent across sources and \emph{i.i.d.} across time. Notice that,  the sleep period parameter $r_l>0$ has been normalized by the mean transmission time $\mathbb{E}[T]$. Let $\mathbf{r}=(r_1, \ldots, r_M)$ be the vector comprising of these sleep period parameters.

\subsection{Total Weighted Average Peak Age }\label{objective_function1}
 Let $U_l(t)$ represent the generation time of the most recently delivered packet from source~$l$ by time $t$. Then, the \emph{age of information}, or simply the \emph{age}, of source~$l$ is defined as \cite{KaulYatesGruteser-Infocom2012} 
\begin{equation}
\Delta_l(t)=t-U_l(t),
\end{equation}
where $\Delta_l(t)$ is right-continuous.
\ifreport As shown in Figure \ref{age_proc}, the age increases linearly with $t$, but is reset to a smaller value upon the delivery of a fresher packet. 
\ifreport 
Observe that a small age $\Delta_l(t)$ indicates that the AP has a fresh status update packet that was generated at source~$l$ recently. Hence, it is desirable to keep $\Delta_l(t)$ small for all the sources.
\fi
\begin{figure}[h]
\includegraphics[scale=0.3]{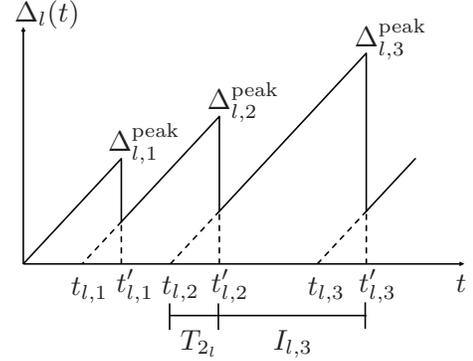}
\centering
\captionsetup{justification=justified}
\caption{The age $\Delta_l(t)$ of source~$l$.
}
\label{age_proc}
\end{figure}

Let us introduce some notations and definitions. Let $i_l$ be the index of the $i$-th delivered packet from source~$l$. We use $t_{l,i}$ and $t_{l,i}'$ to denote the generation and delivery times, respectively, of the $i$-th delivered packet from source~$l$, such that $t_{l,i}'-t_{l,i}=T_{i_l}$.\footnote{A packet of a particular source is deemed delivered when the source receives the feedback.}  Let $I_{l,i}=t_{l,i}'-t_{l,i-1}'$ denote the $i$-th inter-departure time of source~$l$, which satisfies $\mathbb{E}[I_{l,i}]=\mathbb{E}[I_l]$ for all $i$. 
The $i$-th peak age of source~$l$, denoted by $\Delta^{\text{peak}}_{l,i}$, is defined as the AoI of source~$l$ right before the $i$-th packet delivery from source~$l$. As shown in Figure \ref{age_proc}, i.e., we have 
\begin{align}
\Delta^{\text{peak}}_{l,i}=\Delta_l(t_{l,i}'^-),
\end{align}
where $t_{l,i}'^-$ is the time instant just before the delivery time $t_{l,i}'$. 
One can observe from Figure \ref{age_proc} that the peak age is \cite{CostaCodreanuEphremides_TIT}
\begin{align}
\Delta^{\text{peak}}_{l,i}=T_{(i-1)_l}+I_{l,i}.
\end{align}
Hence, the average peak age of source~$l$ is given by 
 \begin{equation}\label{peak_age1}
\mathbb{E}[\Delta^{\text{peak}}_{l}]=\mathbb{E}[T]+\mathbb{E}[I_l],
\end{equation}
where we omit the subscripts $i$ and $i_l$ as $I_{l,i}$'s and $T_{i_l}$'s are  i.i.d. across time.  The average peak age metric provides information regarding the worst case age, with the advantage of having a simpler formulation than the average age metric \cite{CostaCodreanuEphremides_TIT}. Thus, it is suitable for applications that have an upper bound restriction on AoI.

We now derive an expression for $\mathbb{E}[I_l]$.  Let $\alpha_l$ be the probability of the event that the source~$l$ obtains channel access and successfully transmits a packet within a sleep-wake cycle. As shown in \cite{chen2013life}, one can utilize the memoryless property of exponential distributed sleep periods to get
\begin{align}\label{access_prob_in_agiven_cycle}
\alpha_l=\frac{r_l e^{r_l\frac{t_s}{\mathbb{E}[T]}}}{e^{\sum_{i=1}^Mr_i\frac{t_s}{\mathbb{E}[T]}}\sum_{i=1}^Mr_i}.
\end{align}
\ifreport To keep the paper self-contained, we provide the derivation of \eqref{access_prob_in_agiven_cycle} in  Appendix \ref{Appendix_A'}\else For the sake of completeness, we derive the above expression in our technical report \cite[Appendix A]{bedewy2019optimal}\fi. Let $N_l$ denote the total number of sleep-wake cycles between two subsequent successful transmissions of source~$l$.  Because  the probability that source~$l$ obtains channel access and transmits successfully in a given cycle is $\alpha_l$, $N_l$ is geometrically distributed with mean $\frac{1}{\alpha_l}$.  By this and \eqref{access_prob_in_agiven_cycle}, we get
\begin{equation}\label{mean_nf}
\mathbb{E}[N_l]=\frac{e^{\sum_{i=1}^Mr_i\frac{t_s}{\mathbb{E}[T]}}\sum_{i=1}^Mr_i}{r_l e^{r_l\frac{t_s}{\mathbb{E}[T]}}}.
\end{equation}
An inter-departure time duration of  source $l$ is composed of $N_l$ consecutive sleep-wake cycles. With a slight abuse of notation,  let $\textbf{cycle}_{l,k}$ denote the duration of the $k$-th sleep-wake cycle after a successful transmission of source~$l$. Hence, 
\begin{equation}
\mathbb{E}[I_l]=\mathbb{E}\left[\sum_{k=1}^{N_l}\textbf{cycle}_{l,k}\right].
\end{equation}
Note that $\textbf{cycle}_{l,k}$'s are i.i.d. across time. Moreover, since the event $(N_l=n)$ depends only on the history, $N_l$ is a stopping time \cite{shiryaev2007optimal}. Hence, it follows from Wald's identity \cite{wald1973sequential} that
\begin{equation}\label{inter_dep}
\mathbb{E}[I_l]=\mathbb{E}[N_l]\mathbb{E}[\textbf{cycle}],
\end{equation}
where $\mathbb{E}[\textbf{cycle}]$ is the mean duration of a sleep-wake cycle. Each cycle consists of an idle period and a transmission/collision time, see Figure \ref{sleep_wake_cycle}. Using the memoryless property of exponential distribution, we observe that the idle period is the minimum of i.i.d. exponential random variables. Thus, it can be shown that the idle period in each cycle is exponentially distributed with mean value equal to $\mathbb{E}[T]/\sum_{i=1}^Mr_i$, where $\mathbb{E}[T]/r_l$ is the mean of sleep periods of source~$l$. Hence, we have
\begin{equation}\label{swc_mean}
\mathbb{E}[\textbf{cycle}]=\frac{\mathbb{E}[T]}{\sum_{i=1}^Mr_i}+\mathbb{E}[T].
\end{equation}
Substituting the expressions for $\mathbb{E}[N_l]$ and $\mathbb{E}[\textbf{cycle}]$  from \eqref{mean_nf} and  \eqref{swc_mean}, respectively, into \eqref{inter_dep}, and \eqref{peak_age1}, we obtain
\begin{equation}
\begin{split}
\mathbb{E}[\Delta^{\text{peak}}_{l}]=&\frac{e^{-r_l\frac{t_s}{\mathbb{E}[T]}}\mathbb{E}[T]}{r_l}e^{\sum_{i=1}^Mr_i\frac{t_s}{\mathbb{E}[T]}}\left(1+\sum_{i=1}^Mr_i\right)+\mathbb{E}[T].
\end{split}
\end{equation}
In this paper, we aim to minimize the total weighted average peak age, which is given by
\begin{equation}\label{t_avg_peak_age}
\begin{split}
\sum_{l=1}^Mw_l\mathbb{E}[\Delta^{\text{peak}}_{l}]\!\!=&\!\!\sum_{l=1}^M\frac{w_le^{-r_l\frac{t_s}{\mathbb{E}[T]}}\mathbb{E}[T]}{r_l}e^{\sum_{i=1}^Mr_i\frac{t_s}{\mathbb{E}[T]}}\!\!\left(\!1\!\!+\!\!\sum_{i=1}^Mr_i\!\right)\!\\+&\sum_{l=1}^Mw_l\mathbb{E}[T],
\end{split}
\end{equation}
where $w_l>0$ is the weight of source~$l$. These weights enable us to prioritize the sources according to their relative importance \cite{talak2018distributed,talak2018optimizing2}.

\subsection{Energy Constraint}
Each source is equipped with a battery that can possibly be recharged by a renewable energy source, such as solar.  In typical wireless sensor networks, sources have a much smaller power consumption in the sleep mode than in the transmission mode. For example, if the sensor is equipped with the radio unit TR 1000 from RF Monolithic \cite{TR1000,sleep_radio_enrgy_cons}, the power consumption in the sleep mode is 15 $\mu$W while the power consumption in the transmission mode is 24.75 mW. Motivated by this, we assume that the energy dissipation during the sleep mode is negligible as compared to the power consumption in the transmission mode. Moreover, we assume that the sensing time duration $t_s$ is much shorter than the transmission time and hence neglect the energy consumed during channel sensing. In Section \ref{ns3sim}, we show that these assumptions have a negligible effect on the performance of the proposed sleep-wake scheduling algorithm. Under these assumptions, the amount of energy used by a source is  equal to the amount of energy consumed in packet transmissions and  feedback receptions. 

The energy constraint on source $l$ is described by the following parameters: a) Initial battery level $B_{l}$,  which denotes the initial amount of energy stored in the battery, b) Target lifetime $D_{l}$,  which is the minimum time-duration that the source~$l$ should be active before its battery is depleted, c) Average energy replenishment rate\footnote{It is assumed that $R_l$ is either known, or it can be estimated accurately.} $R_l$, which is the rate at which the battery of source~$l$ receives energy from its energy source. If source~$l$ does not have access to an energy source, then we have $R_l=0$.  Define $P_{\text{max},l}$ for source $l$ as
\begin{equation}\label{max_energ}
P_{\text{max},l}=\frac{B_l}{D_l}+R_l, ~\forall l,
\end{equation}
where $P_{\text{max},l}$  is the maximum allowable power consumption of source $l$ such that the target lifetime $D_l$ is met.

For the sleep-wake scheduling mechanism under consideration, it has been shown in \cite{chen2013life} that the  fraction of time in which source~$l$ is in the transmission mode is given by
\begin{equation}\label{sigma_l}
\sigma_l=\frac{[1-e^{-r_l\frac{t_s}{\mathbb{E}[T]}}]\sum_{i=1}^Mr_i+r_le^{-r_l\frac{t_s}{\mathbb{E}[T]}}}{\sum_{i=1}^Mr_i+1}.
\end{equation}
For the sake of completeness, the derivation of $\sigma_l$ is provided in \ifreport Appendix \ref{Appendix_A''}\else our technical report \cite[Appendix B]{bedewy2019optimal}\fi. Let $P_{\text{avg},l}$ denote the average power consumption of source~$l$ in the transmission mode. Then the actual power consumption  of source~$l$, denoted by $P_{\text{act},l}$, is given by
\begin{equation}\label{eq14}
P_{\text{act},l}=\sigma_lP_{\text{avg},l},~ \forall l.
\end{equation}
 For source $l$ to achieve its target lifetime $D_l$, we must have 
\begin{equation}\label{energy_const_1}
P_{\text{act},l}\leq P_{\text{max},l}, ~\forall l.
\end{equation}

Define $b_l\triangleq P_{\text{max},l}/P_{\text{avg},l}$ as the target power efficiency of source~$l$. By using \eqref{sigma_l}-\eqref{eq14},  the constraints in \eqref{energy_const_1} can be rewritten as
\begin{equation}\label{energy_const_2}
\sigma_l=\frac{[1-e^{-r_l\frac{t_s}{\mathbb{E}[T]}}]\sum_{i=1}^Mr_i+r_le^{-r_l\frac{t_s}{\mathbb{E}[T]}}}{\sum_{i=1}^Mr_i+1}\leq b_l, ~\forall l.
\end{equation}
Because $\sigma_l \leq 1$, if  $b_l\geq 1$, then constraint \eqref{energy_const_2} is always satisfied. 

\subsection{Problem Formulation}
Our goal is to find the optimal  sleep-wake parameters $\mathbf{r}$ that minimizes the total weighted average peak age in \eqref{t_avg_peak_age}, while simultaneously ensuring  the energy constraints \eqref{energy_const_2} for all sources. Dividing the objective function \eqref{t_avg_peak_age} by  $\mathbb{E}[T]$, we obtain the following optimization problem: 
(Problem \textbf{1})
\begin{align}\label{problem1}
\begin{split}
\!\!\!\!\bar{\Delta}^{\text{w-peak}}_{\text{opt}}\triangleq\min_{r_l>0}& \sum_{l=1}^M\frac{w_le^{-r_l\frac{t_s}{\mathbb{E}[T]}}}{r_l}e^{\sum_{i=1}^Mr_i\frac{t_s}{\mathbb{E}[T]}}\left(1+\sum_{i=1}^Mr_i\right)+\\&\sum_{l=1}^Mw_l\\
\textbf{s.t.}~&\frac{[1-e^{-r_l\frac{t_s}{\mathbb{E}[T]}}]\sum_{i=1}^Mr_i+r_le^{-r_l\frac{t_s}{\mathbb{E}[T]}}}{\sum_{i=1}^Mr_i+1}\leq b_l,\forall l,
\end{split}\!\!\!\!\!\!\!\!\!\!\!
\end{align}
where $\bar{\Delta}^{\text{w-peak}}_{\text{opt}}$ is the optimal objective value of Problem \textbf{1}. We will use $\bar{\Delta}^{\text{w-peak}}(\mathbf{r})$ to denote the objective value for given sleeping period parameters $\mathbf{r}$. One can notice from \eqref{problem1} that the optimal sleeping period parameters depend on the sensing time $t_s$ and the mean transmission time $\mathbb{E}[T]$ only through their ratio $t_s/\mathbb{E}[T]$. This insight plays a crucial role in subsequent analysis of Problem \textbf{1}.



%% file: sections/main_result.tex
\section{Main Results}\label{main_result}
When $t_s = 0$, although Problem \textbf{1} is non-convex, it can be solved by defining an auxiliary variable $y=\sum_{i=1}^Mr_i+1$ and applying a nested optimization algorithm: In the inner layer, we optimize $r_l$ for a given $y$. Then, we write the optimized objective as a function of $y$. In the outer layer, we optimize $y$. It happens that the inner and outer layer optimization problems are both convex. The details can be found in Section \ref{discussion}.

 However, this method does not work for positive sensing times $t_s>0$ and Problem \textbf{1} becomes non-convex. Hence, it is challenging to optimize $\mathbf{r}$ for positive $t_s$. In this section, we develop a low-complexity closed-form solution which is shown to be  near-optimal if the sensing time $t_s$ is short as compared with the mean transmission time $\mathbb{E}[T]$. Our solution is developed by considering the following two regimes separately: (i) \emph{Energy-adequate regime} denoted as $\sum_{i=1}^M b_i \geq 1$, where the condition $\sum_{i=1}^M b_i \geq 1$ means that the sources have a sufficient amount of total energy to ensure that at least one source is awake at any time, (ii) \emph{Energy-scarce regime} represented by $\sum_{i=1}^M b_i < 1$, which indicates that the sources have to sleep for some time to meet the sources' energy constraints.
\subsection{Energy-adequate Regime}
In the energy-adequate regime $\sum_{i=1}^M b_i \geq 1$, our solution $\mathbf{r}^{\star}:= (r^{\star}_1,\ldots,r^{\star}_M)$ is given as 
\begin{equation}\label{r_f_b_gr_1}
r^{\star}_l=\min\{b_l,\beta^{\star}\sqrt{w_l}\} x^{\star},\forall l,
\end{equation}
where $x^\star$ and $\beta^\star$ are expressed in terms of the parameters $\{b_i,w_i\}_{i=1}^{M},t_s/\mathbb{E}[T] $
as follows:
\begin{equation}\label{x*}
x^{\star}=\frac{-1}{2}+\sqrt{\frac{1}{4}+\frac{\mathbb{E}[T]}{t_s}},
\end{equation}
and $\beta^{\star}$ is the unique  root of 
\begin{equation}\label{condition_beta}
\sum_{i=1}^M\min\{b_i,\beta^{\star}\sqrt{w_i}\}\ =1.
\end{equation}
The performance of the above solution $\mathbf{r}^\star$ is manifested in the following theorem:
%
%
%
\begin{theorem}[\textbf{Near-optimality}]\label{thm_case_sum_bi_ge_1}
If $\sum_{i=1}^Mb_i\ge 1$, then the solution $\mathbf{r}^{\star}$  \eqref{r_f_b_gr_1} - \eqref{condition_beta} is near-optimal for solving \eqref{problem1} when $t_s/E[T]$ is sufficiently small, in the following sense:\footnote{We use the standard order notation: $f(h)=O(g(h))$ means $z_1\le \lim_{h\to 0}f(h)/g(h)\le z_2$ for some constants $z_1>0$ and $z_2>0$, while $f(h)=o(g(h))$ means $\lim_{h\to 0}f(h)/g(h)= 0$.}
\begin{align}\label{sub_opt_gap_eq}
\left|\bar{\Delta}^{\text{w-peak}}(\mathbf{r}^{\star})-\bar{\Delta}^{\text{w-peak}}_{\text{opt}}\right| \leq 2\sqrt{\frac{t_s}{\mathbb{E}[T]}}C_1\!+\!o\left(\sqrt{\frac{t_s}{\mathbb{E}[T]}}\right),
\end{align}
where
\begin{align}\label{eq23thm3_1}
C_1=\sum_{i=1}^M\frac{w_i}{\min\{b_i,\beta^{\star}\sqrt{w_i}\}}.
\end{align}
\end{theorem}
 \ifreport
\begin{proof}
See Section \ref{proof_case_sum_bi_ge_1}. 
\end{proof}
\else 
\begin{proof}
See Section \ref{proof1}.
\end{proof}
\fi
From Theorem \ref{thm_case_sum_bi_ge_1}, we can obtain the following corollary:
\begin{corollary}[\textbf{Asymptotic optimality}]\label{cor1}
If $\sum_{i=1}^Mb_i\ge 1$, then the solution $\mathbf{r}^{\star}$ \eqref{r_f_b_gr_1} - \eqref{condition_beta} is asymptotically optimal for Problem \textbf{1} in \eqref{problem1} as $t_s/\mathbb{E}[T]\rightarrow 0$, i.e., 
\begin{align}\label{zero_gab_result}
\lim_{\frac{t_s}{\mathbb{E}[T]}\rightarrow 0} \left|\bar{\Delta}^{\text{w-peak}}(\mathbf{r}^{\star})-\bar{\Delta}^{\text{w-peak}}_{\text{opt}}\right|=0.
\end{align}
Moreover, the asymptotic optimal objective value of Problem \textbf{1} as $t_s/\mathbb{E}[T]\to 0$ is\ifreport\footnote{Observe that, according to \eqref{asymptotic_value_final}, the asymptotic optimal average peak age of
source $l$ is $(1/\min\{b_l,\beta^\star\sqrt{w_l}\} + 1)$ which decreases with the weight $w_l$. The weighted average peak age is $w_l(1/\min\{b_l,\beta^\star\sqrt{w_l}\} + 1)$ which increases with $w_l$. This phenomenon is reasonable and agrees with our expectation.}\fi
  \begin{equation}\label{asymptotic_value_final}
\lim_{\frac{t_s}{\mathbb{E}[T]}\rightarrow 0} \bar{\Delta}^{\text{w-peak}}_{\text{opt}}=\sum_{i=1}^M\left[\frac{w_i}{\min\{b_i,\beta^{\star}\sqrt{w_i}\}}+w_i\right].
\end{equation}
 \end{corollary}
 \ifreport
\begin{proof}
See Section \ref{proof_case_sum_bi_ge_1}. 
\end{proof}
\else 
\begin{proof}
See Section \ref{proof1}.
\end{proof}
\fi 
 \subsection{Energy-scarce Regime}\label{case_sum_bi_le_1}
Now, we present a solution to Problem \textbf{1}  in the energy-scarce regime $\sum_{i=1}^Mb_i< 1$, and show it is near-optimal. The solution $\mathbf{r}^\star$ of the energy-scarce regime is again given by \eqref{r_f_b_gr_1}, where $x^\star$ and $\beta^\star$ are
\begin{equation}\label{condition_x*_beta*}
x^{\star}=\frac{\min_l c_l}{1-\sum_{i=1}^Mb_i},~\beta^{\star}=\sum_{i=1}^M\frac{1}{\sqrt{w_i}},
\end{equation}
and 
\begin{align}
c_l =&\frac{2b_l\left(1-\sum_{i=1}^Mb_i\right)^2}{Q_l},\label{feasible_factor} \\
\begin{split}
Q_l= &b_l\left(\!1\!-\!\sum_{i=1}^Mb_i\!\right)^2\\
&+\!\!\sqrt{b_l^2\!\left(\!1\!-\!\sum_{i=1}^Mb_i\!\right)^4\!\!\!\!+\!4b_l^2\left(\!1\!-\!\sum_{i=1}^Mb_i\!\right)^2\!\!\left(\!\sum_{i=1}^Mb_i\!-\!b_l\!\right)\!\frac{t_s}{\mathbb{E}[T]}}.\label{feas_fact_eq27}
\end{split}
\end{align}
The near-optimality of the proposed solution (i.e., $\mathbf{r}^{\star}$) in the energy scarce regime is explained in the following theorem:
\begin{theorem}[\textbf{Near-optimality}]\label{thm_case_sum_bi_le_1}
If $\sum_{i=1}^Mb_i< 1$, then  the solution $\mathbf{r}^{\star}$  \eqref{r_f_b_gr_1} and  \eqref{condition_x*_beta*} - \eqref{feas_fact_eq27} is near-optimal for solving \eqref{problem1} when $t_s/\mathbb{E}[T]$ is sufficiently small, in the following sense:
\begin{align}\label{sub_opt_gap_eq_2}
\left|\bar{\Delta}^{\text{w-peak}}(\mathbf{r}^{\star})-\bar{\Delta}^{\text{w-peak}}_{\text{opt}}\right| \leq \frac{t_s}{\mathbb{E}[T]}C_2\!+\!o\left(\frac{t_s}{\mathbb{E}[T]}\right),
\end{align}
where
\begin{align}\label{eq29thm3_3}
C_2=\sum_{l=1}^M\frac{w_l}{b_l(1-\sum_{i=1}^Mb_i)}\left(3\sum_{i=1}^Mb_i-\min_j b_j\right).
\end{align}
\end{theorem}
 \ifreport
\begin{proof}
See Section \ref{proof_case_sum_bi_le_1}. 
\end{proof}
\else 
\begin{proof}
See our technical report \cite{bedewy2019optimal}.
\end{proof}
\fi
We obtain the following corollary from Theorem \ref{thm_case_sum_bi_le_1}.
\begin{corollary}[\textbf{Asymptotic optimality}]\label{cor2}
If $\sum_{i=1}^Mb_i< 1$, then \eqref{zero_gab_result} holds for the solution $\mathbf{r}^{\star}$ \eqref{r_f_b_gr_1} and \eqref{condition_x*_beta*} - \eqref{feas_fact_eq27}. In other words, our proposed solution is asymptotically optimal for Problem \textbf{1}  in \eqref{problem1}  as $t_s/\mathbb{E}[T]\rightarrow 0$. Moreover, the asymptotic optimal objective value of Problem \textbf{1} as $t_s/\mathbb{E}[T]\to 0$ is 
\begin{equation}\label{asymptotic_value_final_2}
\begin{split}
\lim_{\frac{t_s}{\mathbb{E}[T]}\rightarrow 0} \bar{\Delta}^{\text{w-peak}}_{\text{opt}}&=\sum_{i=1}^M\left[\frac{w_i}{\min\{b_i,\beta^{\star}\sqrt{w_i}\}}+w_i\right]\\&
=\sum_{i=1}^M\left[\frac{w_i}{b_i}+w_i\right].
\end{split}
\end{equation}
\end{corollary}
 \ifreport
\begin{proof}
See Section \ref{proof_case_sum_bi_le_1}. 
\end{proof}
\else 
\begin{proof}
See our technical report \cite{bedewy2019optimal}.
\end{proof}
\fi
Interestingly,  the asymptotic optimal objective values of Problem \textbf{1} in both regimes, given by \eqref{asymptotic_value_final} and \eqref{asymptotic_value_final_2}, are of an identical expression. However, in the energy-scarce regime, we can observe that $\beta^{\star}$, which is defined in \eqref{condition_x*_beta*}, always satisfies $\min\{b_l,\beta^{\star}\sqrt{w_l}\}=b_l$ for all $l$.
\begin{remark}
We would like to point out that the condition $t_s/\mathbb{E}[T]\approx 0$ is satisfied in many practical applications. For instance, in a wireless sensor network that is equipped with low-power UHF transceivers \cite{el2002spatial}, the carrier  sensing   time   is $t_s=40~\mu$s, while the transmission time is around $5$ ms. Hence, $t_s/\mathbb{E}[T]\approx 0.008$.
\end{remark}
\subsection{Discussion}\label{discussion}
In this subsection, we present a simple implementation of our proposed solution, discuss the nested convex optimization method that can be used to solve Problem \textbf{1} when $t_s=0$,  provide some useful insights about our proposed solution at the limit point $t_s/\mathbb{E}[T]\to 0$, and provide a comparison with synchronized schedulers performance. \ifreport \else Due to space limitation, we move the second point to our technical report \cite[Section 3.3.2]{bedewy2019optimal}.\fi  

\subsubsection{Implementation of Sleep-wake Scheduling}
We devise a simple algorithm to compute our solution $\mathbf{r}^\star$, which is provided in Algorithm \ref{alg1}. Notice that $\mathbf{r}^\star$ has the same expression \eqref{r_f_b_gr_1} in the energy-adequate and energy-scarce regimes. We exploit this fact to simplify the implementation of sleep-wake scheduling. In particular, the sources report $w_l$ and $b_l$ to the AP, which computes $\beta^\star$ and $x^\star$, and broadcasts them back to the sources. After receiving $\beta^\star$ and $x^\star$,  source~$l$ computes $r_l^\star$ based on \eqref{r_f_b_gr_1}. In practical wireless sensor networks, e.g., smart city networks and industrial control sensor networks \cite{wsn_industry,hsieh2018decentralized}, the sensors report their measurements via an access point (AP). Hence, it is reasonable to employ the AP in implementing the sleep-wake scheduler.
\begin{algorithm}[h]
\SetKwData{NULL}{NULL}
The AP gathers the parameters  $\{(w_i,b_i)_{i=1}^M,t_s/\mathbb{E}[T]\}$\;
\uIf{$\sum_{i=1}^Mb_i\ge 1$}{
The AP computes $x^{\star},\beta^{\star}$ from \eqref{x*} and \eqref{condition_beta}\;}
\Else{
The AP computes  $x^{\star},\beta^{\star}$ from  \eqref{condition_x*_beta*} - \eqref{feas_fact_eq27}\;}
The AP broadcasts $x^{\star},\beta^{\star}$ to all the $M$ sources\;
Upon hearing $x^{\star},\beta^{\star}$, source~$l$ compute $r^{\star}_l$ from \eqref{r_f_b_gr_1}\;
\caption{Implementation of sleep-wake scheduler.}\label{alg1}
\end{algorithm}

In the above implementation procedure, the sources do not need to know if the overall network is in the energy-adequate or energy-scarce regime; only the AP knows about it. Further, the amount of downlink signaling overhead is small, because only two parameters $\beta^\star$ and $x^\star$ are broadcasted to the sources. Moreover, when the node density is high, the scalability of the network is a crucial concern and reporting $w_l$ and $b_l$ for each source is impractical. In this case, the AP can compute $\beta^\star$ and $x^\star$ by estimating the distribution of $w_l$ and $b_l$, as well as the number of source nodes, which reduces the uplink signaling overhead. Finally, when sources are not in the hearing range of each other, hidden/exposed source problems arise. These problems are challenging to solve analytically. However, this can be solved by designing practical heuristic solutions based on the theoretical solutions. One design method was given in  \cite{chen2013life}.

\ifreport
\subsubsection{The Nested Convex Optimization Method for $t_s=0$}
If  $t_s=0$, Problem \textbf{1} reduces to the following optimization problem:
\begin{align}\label{problem1_ts_0}
\begin{split}
\bar{\Delta}^{\text{w-peak}}_{\text{opt}}\triangleq\min_{r_l>0}& \sum_{l=1}^M\frac{w_l\left(1+\sum_{i=1}^Mr_i\right)}{r_l}+\sum_{l=1}^Mw_l\\
\textbf{s.t.}~&r_l\leq b_l(\sum_{i=1}^Mr_i+1),\forall l.
\end{split}
\end{align}
Observe that the optimization problem in \eqref{problem1_ts_0} is non-convex. To bypass this difficulty, we use an auxiliary variable $y=\sum_{i=1}^Mr_i+1$. Hence, we obtain the following optimization problem for given $y$:
\begin{align}
\begin{split}\label{nested1}
\min_{r_i>0}& \sum_{i=1}^M\left[\frac{w_iy}{r_i}+w_i\right]
\end{split}\\
\textbf{s.t.}~&r_l\leq b_ly, \forall l,\label{nested2}\\&
\sum_{i=1}^Mr_i+1=y.\label{nested3}
\end{align}
The objective function in \eqref{nested1} is a convex function. Moreover, the constraints in \eqref{nested2} and \eqref{nested3} are affine. Hence, Problem \eqref{nested1} is convex. Exploiting  \eqref{nested1},  we solve \eqref{problem1_ts_0} by using
a two-layer nested convex optimization method:  In the inner layer, we optimize $\mathbf{r}$ for  given $y$. After solving $\mathbf{r}$, we will optimize $y$ in the outer layer. This technique is used in the proof of Lemma \ref{lemma_bounds} in Appendix \ref{Appendix_B}, where the reader can find the detailed solution.
\fi

\subsubsection{Asymptotic Behavior of The Optimal Solution}
\ifreport In the energy-adequate regime, the sleeping period parameter $r_l^\star$ of source~$l$ tends to infinity as $t_s/\mathbb{E}[T] \to 0$, while the ratio $r^{\star}_l/r^{\star}_i$ between source~$l$ and source~$i$ is kept as a constant for all $l$ and $i$. Hence, the sleeping time of the sources tends to zero. Meanwhile, since $t_s/ \mathbb{E}[T] \to 0$, the sensing time becomes negligible. The channel access probability of source~$l$ in this limit can be computed as 
\begin{align}\label{sol_sum_bi_at_limit}
\lim_{\frac{t_s}{\mathbb{E}[T]}\to 0}\sigma^{\star}_l=\min\{b_l,\beta^{\star}\sqrt{w_l}\}. 
\end{align} 
Because of \eqref{condition_beta}, $\lim_{t_s/\mathbb{E}[T]\to 0} \sum_{i=1}^M\sigma_i^\star = 1$. Hence, the channel is occupied by the sources at all time, without any time overhead spent on sensing and sleeping. 

On the other hand, in the energy-scarce regime, the sleeping period parameter $r^{\star}_l$ of source~$l$ converges to a constant value when $t_s/\mathbb{E}[T]\to 0$, i.e., we have
\begin{align}
 \lim_{\frac{t_s}{\mathbb{E}[T]}\to 0}r^{\star}_l=\frac{b_l}{1-\sum_{i=1}^Mb_i}.
\end{align}
 Since the cumulative energy is scarce, the sources necessarily need to stay idle for some time in order to meet their target lifetime. Hence, sleep periods are imposed for achieving the optimal trade-off between minimizing AoI and energy consumption. \fi

\subsubsection{Comparison with Synchronized Schedulers Performance} 
 We would like to show that the performance of our proposed algorithm is asymptotically no worse than any synchronized (e.g., centralized) scheduler.
Consider a  scheduler in which the fraction of time during which source~$l$ transmits update packets is equal to $a_l$, where we have 
$\mathbf{a} = \left\{ a_l \right\}_{l=1}^M$ and $ \sum_{i=1}^M a_i \leq 1$. 
In this scheduler, only one source is allowed to access the channel at a time, i.e., there is no collision (this can be achieved either by a deterministic scheduler or by assigning a channel access probability $a_l$ for each source $l$ after each packet transmission)\footnote{Note that if $\sum_{i=1}^M a_i < 1$, then it is possible that the scheduler decides not to serve any source after the transmission of some packet. In this case, the scheduler waits for a random time that has the same distribution as the transmission time $T$ before deciding to serve another source.}. We can perform an analysis similar to that of Section~\ref{objective_function1}, and show that the total weighted average peak age of a synchronized scheduler is given by
\begin{equation}\label{obj_func_c}
\sum_{i=1}^M\left[\frac{w_i\mathbb{E}[T]}{a_i}+w_i~\mathbb{E}[T]\right].
\end{equation}
Hence, the problem of designing an optimal synchronized scheduler that minimizes the total weighted average peak age under energy constraints can be cast as
\begin{align}
\bar{\Delta}_{\text{opt-s}}^{\text{w-peak}}\triangleq\min_{a_i>0}&\sum_{i=1}^M\left[\frac{w_i}{a_i}+w_i\right]\label{problem_centr}\\
\textbf{s.t.}~&a_l\le b_l,~\forall l,\label{eq105}\\&
\sum_{i=1}^Ma_i\leq1,\label{eq106}
\end{align}
where we have divided the objective function by $\mathbb{E}[T]$. Next, we show that the performance of our proposed algorithm converges to that of the optimal synchronized scheduler when $t_s/\mathbb{E}[T]\to 0$. 
\begin{corollary}\label{corollary_centeralized_scheduler}
For any $(w_i,b_i)_{i=1}^M$, we have
\begin{align}
\lim_{\frac{t_s}{\mathbb{E}[T]}\rightarrow 0} \bar{\Delta}^{\text{w-peak}}_{\text{opt}}=\bar{\Delta}_{\text{opt-s}}^{\text{w-peak}}.
\end{align}
\end{corollary}
\ifreport
\begin{proof}
The proof is provided in Appendix \ref{Appendix_E} which is listed at the end  before Appendix \ref{sec:learn_proof} as it requires some results from precedent appendixes.
\end{proof}
\else
\begin{proof}
See our technical report \cite{bedewy2019optimal}.
\end{proof}
\fi
Synchronized schedulers were recently studied in \cite{talak2018optimizing2} for the case without energy constraints, i.e., $b_l \geq 1$ for all $l$. According to Corollary \ref{corollary_centeralized_scheduler}, the channel access probability of the synchronized scheduler in \cite{talak2018optimizing2} \ifreport is a special case of our solution \eqref{sol_sum_bi_at_limit} where $b_l \geq 1$ for all $l$.\else can be obtained from our solution.\fi

%% file: sections/proof_main_result.tex
\section{Proofs of the Main Results}\label{proof1}
\ifreport
In this section, we provide the proofs of Theorem \ref{thm_case_sum_bi_ge_1}, Corollary \ref{cor1}, Theorem \ref{thm_case_sum_bi_le_1}, and Corollary \ref{cor2}.
\subsection{The Proofs of Theorem \ref{thm_case_sum_bi_ge_1} and Corollary \ref{cor1}}\label{proof_case_sum_bi_ge_1}
We prove Theorem \ref{thm_case_sum_bi_ge_1} and Corollary \ref{cor1} in three steps: 

\textbf{Step 1}: We show that our solution $\mathbf{r}^{\star}$  \eqref{r_f_b_gr_1} - \eqref{condition_beta} is feasible for Problem \textbf{1}. 
\begin{lemma}\label{lemma_feasibility}
If $\sum_{i=1}^Mb_i\ge 1$, then the solution $\mathbf{r}^{\star}$   \eqref{r_f_b_gr_1} - \eqref{condition_beta} is feasible for Problem \textbf{1}. 
\end{lemma}
 \ifreport
\begin{proof}
See Appendix \ref{Appendix_A}.
\end{proof}
\else 
\begin{proof}
See our technical report \cite{bedewy2019optimal}.
\end{proof}
\fi
Hence, by substituting this solution $\mathbf{r}^{\star}$ into the objective function of Problem \textbf{1} in \eqref{problem1}, we get an upper bound on the optimal value $\bar{\Delta}^{\text{w-peak}}_{\text{opt}}$, which is expressed in the following lemma: 
\begin{lemma}\label{lemma_upper_bound_sum_i_b_i_ge1}
If $\sum_{i=1}^Mb_i\ge 1$, then
\begin{align}\label{upper_bound_sumb_i_ge_1}
\bar{\Delta}^{\text{w-peak}}_{\text{opt}}\le \bar{\Delta}^{\text{w-peak}}(\mathbf{r}^{\star})\le \sum_{i=1}^M\left[\frac{w_ie^{x^{\star}\frac{t_s}{\mathbb{E}[T]}}\left(1+\frac{1}{x^{\star}}\right)}{\min\{b_i,\beta^{\star}\sqrt{w_i}\}}+w_i\right],
\end{align}
where $x^{\star}$, $\beta^{\star}$ are defined in \eqref{x*}, \eqref{condition_beta}. 
\end{lemma}
\begin{proof}
In Lemma \ref{lemma_feasibility}, we showed that our proposed solution $\mathbf{r}^{\star}$  \eqref{r_f_b_gr_1} - \eqref{condition_beta} is feasible for Problem \textbf{1}. Hence, we substitute this solution into Problem \textbf{1} to obtain the following upper bound:
\begin{align}
\sum_{i=1}^M\left[\frac{w_ie^{x^{\star}\frac{t_s}{\mathbb{E}[T]}}\left(1+\frac{1}{x^{\star}}\right)e^{-\min\{b_i,\beta^{\star}\sqrt{w_i}\}x^\star\frac{t_s}{\mathbb{E}[T]}}}{\min\{b_i,\beta^{\star}\sqrt{w_i}\}}+w_i\right].
\end{align}
Next, we replace $e^{-\min\{b_i,\beta^{\star}\sqrt{w_i}\}x^\star(t_s/\mathbb{E}[T])}$ by 1 
to derive another upper bound with a simple expression, which is given by \eqref{upper_bound_sumb_i_ge_1}. This completes the proof. 
\end{proof}

\textbf{Step 2}: 
We now construct a lower bound on the optimal
value of Problem \textbf{1}. Suppose that $\mathbf{r}=(r_1,\ldots,r_M)$ is a feasible solution to Problem \textbf{1}, such that $r_l >0$ and
\begin{align}
\frac{[1-e^{-r_l\frac{t_s}{\mathbb{E}[T]}}]\sum_{i=1}^Mr_i+r_le^{-r_l\frac{t_s}{\mathbb{E}[T]}}}{\sum_{i=1}^Mr_i+1}\leq b_l, \forall l.
\end{align} 
Because $[1-e^{-r_l(t_s/\mathbb{E}[T])}]\sum_{i=1}^Mr_i+r_le^{-r_l(t_s/\mathbb{E}[T])} > r_l$ for all $l$, $\mathbf{r}$ satisfies $r_l/(\sum_{i=1}^Mr_i+1)\leq b_l$. Hence, the following Problem \textbf{2} has a larger feasible set than Problem \textbf{1}: (Problem \textbf{2})
\begin{align}
\begin{split}\label{problem2}
\bar{\Delta}^{\text{w-peak}}_{\text{opt},2}\triangleq\min_{r_l>0}& \sum_{l=1}^M\frac{w_le^{-r_l\frac{t_s}{\mathbb{E}[T]}}}{r_l}e^{\sum_{i=1}^Mr_i\frac{t_s}{\mathbb{E}[T]}}\left(1+\sum_{i=1}^Mr_i\right)\\&+\sum_{l=1}^Mw_l
\end{split}\!\!\!\!\!\!\!\!\!\\
\begin{split}\label{constrain_final_1}
\textbf{s.t.}~&r_l\leq b_l\left(\sum_{i=1}^Mr_i+1\right), \forall l,
\end{split}
\end{align}
where $\bar{\Delta}^{\text{w-peak}}_{\text{opt},2}$ is the optimal value of Problem \textbf{2}. The optimal objective value of Problem \textbf{2} is a lower bound of that of Problem \textbf{1}. We note that the constraint set corresponding to Problem~\textbf{2} is convex. Thus, this relaxation converts the constraint set of Problem~\textbf{1} to a convex one, and hence enables us to obtain a lower bound for the optimal value of Problem \textbf{1}, which is expressed in the following lemma: 
\begin{lemma}\label{lemma_bounds}
If $\sum_{i=1}^Mb_i\ge 1$, then  
\begin{align}
\bar{\Delta}^{\text{w-peak}}_{\text{opt}}\ge\bar{\Delta}^{\text{w-peak}}_{\text{opt},2}\ge\sum_{i=1}^M\left[\frac{w_i}{\min\{b_i,\beta^{\star}\sqrt{w_i}\}}+w_i\right],\label{lower_bound_sum_bi_ge_1}
\end{align}
where $\beta^{\star}$ is the root of \eqref{condition_beta}. 
\end{lemma}
\ifreport
\begin{proof}
See Appendix \ref{Appendix_B}.
\end{proof}
\else 
\begin{proof}
See our technical report \cite{bedewy2019optimal}.
\end{proof}
\fi

\textbf{Step 3}: After the upper and lower bounds of $\bar{\Delta}^{\text{w-peak}}_{\text{opt}}$ were derived in Steps 1-2, we are ready to analysis their gap. By combining \eqref{upper_bound_sumb_i_ge_1} and \eqref{lower_bound_sum_bi_ge_1}, the sub-optimality gap of the solution $\mathbf{r}^\star$  \eqref{r_f_b_gr_1} - \eqref{condition_beta} is upper bounded by
\begin{align}\label{G1}
\begin{split}
&\left|\bar{\Delta}^{\text{w-peak}}(\mathbf{r}^{\star})\!-\!\bar{\Delta}^{\text{w-peak}}_{\text{opt}}\right|\le \sum_{i=1}^M\frac{w_i\left(e^{x^{\star}\frac{t_s}{\mathbb{E}[T]}}(1\!+\!\frac{1}{x^{\star}})\!-\!1 \right)}{\min\{b_i,\beta^{\star}\sqrt{w_i}\}},
\end{split}
\end{align}
where $x^{\star}$, $\beta^{\star}$ are defined in \eqref{x*}, \eqref{condition_beta}. Next, we characterize the right-hand-side (RHS) of \eqref{G1} by Taylor expansion. For simplicity, let $\epsilon=\frac{t_s}{\mathbb{E}[T]}$. Using the expression for $x^{\star}$ from \eqref{x*}, we have
\begin{equation}\label{x^*t_s}
\begin{split}
x^{\star}\epsilon=&-\frac{\epsilon}{2}+\sqrt{\frac{\epsilon^2}{4}+\epsilon}
=\frac{\epsilon}{\frac{\epsilon}{2}+\sqrt{\frac{\epsilon^2}{4}+\epsilon}}
=\sqrt{\epsilon}+o(\sqrt{\epsilon}).
\end{split}
\end{equation}
Moreover,
\begin{equation}\label{x^*ET}
\begin{split}
x^\star=&-\frac{1}{2}+\sqrt{\frac{1}{4}+\frac{1}{\epsilon}}
=\frac{\frac{1}{\epsilon}}{\frac{1}{2} + \sqrt{\frac{1}{4} + \frac{1}{\epsilon}}}
=\frac{1}{\sqrt{\epsilon}} + o\left(\frac{1}{\sqrt{\epsilon}}\right). 
\end{split}
\end{equation}
 Substituting \eqref{x^*t_s} and~\eqref{x^*ET} in \eqref{G1}, we obtain
\begin{align}
&\left|\bar{\Delta}^{\text{w-peak}}(\mathbf{r}^{\star})-\bar{\Delta}^{\text{w-peak}}_{\text{opt}}\right|\\
&\le \sum_{i=1}^M\frac{w_i[e^{\sqrt{\epsilon} + o(\sqrt{\epsilon})} (1 + \sqrt{\epsilon} + o(\sqrt{\epsilon})) - 1 ]}{\min\{b_i,\beta^{\star}\sqrt{w_i}\}}\nonumber\\&
=\sum_{i=1}^M\frac{w_i[(1\!+\! \sqrt{\epsilon}\! +\!o(\sqrt{\epsilon}))(1 \!+\! \sqrt{\epsilon} \!+\! o(\sqrt{\epsilon})) \!-\! 1 ]}{\min\{b_i,\beta^{\star}\sqrt{w_i}\}}\nonumber\\&
=2 \sqrt{\epsilon} \sum_{i=1}^M\frac{w_i}{\min\{b_i,\beta^{\star}\sqrt{w_i}\}}  + o(\sqrt{\epsilon}),
\end{align}
where the second inequality involves the use of Taylor expansion. This proves Theorem \ref{thm_case_sum_bi_ge_1}. 

We can observe that the gap $\left|\bar{\Delta}^{\text{w-peak}}(\mathbf{r}^{\star})-\bar{\Delta}^{\text{w-peak}}_{\text{opt}}\right|$ in the energy-adequate regime converges to zero at a speed of $O(\sqrt{\epsilon})$, as $\epsilon \to 0$. Further,  both the upper and lower bounds \eqref{upper_bound_sumb_i_ge_1}, \eqref{lower_bound_sum_bi_ge_1}, converge to $\sum_{i=1}^M[(w_i/\min\{b_i,\beta^{\star}\sqrt{w_i}\})+w_i]$ as $t_s/\mathbb{E}[T]\rightarrow 0$. Thus, this value is the asymptotic optimal objective value of Problem \textbf{1}. This proves Corollary \ref{cor1}. 

\else
In this section, we provide the proofs of Theorem \ref{thm_case_sum_bi_ge_1} and Corollary \ref{cor1}. The proofs of  Theorem \ref{thm_case_sum_bi_le_1} and Corollary \ref{cor2} have the same idea, and are provided in our technical report \cite{bedewy2019optimal}. We prove Theorem \ref{thm_case_sum_bi_ge_1} and Corollary \ref{cor1} in three steps:

\textbf{Step 1}: We begin by showing that our solution $\mathbf{r}^{\star}$  \eqref{r_f_b_gr_1} - \eqref{condition_beta} is feasible for Problem \textbf{1}. 
\begin{lemma}\label{lemma_feasibility}
If $\sum_{i=1}^Mb_i\ge 1$, then the solution $\mathbf{r}^{\star}$   \eqref{r_f_b_gr_1} - \eqref{condition_beta} is feasible for Problem \textbf{1}. 
\end{lemma} 
\begin{proof}
See our technical report \cite{bedewy2019optimal}.
\end{proof}
Hence, by substituting this solution $\mathbf{r}^{\star}$ into the objective function of Problem \textbf{1} in \eqref{problem1}, we get an upper bound on the optimal value $\bar{\Delta}^{\text{w-peak}}_{\text{opt}}$, which is expressed in the following lemma: 
\begin{lemma}\label{lemma_upper_bound_sum_i_b_i_ge1}
If $\sum_{i=1}^Mb_i\ge 1$, then
\begin{align}\label{upper_bound_sumb_i_ge_1}
\bar{\Delta}^{\text{w-peak}}_{\text{opt}}\le \bar{\Delta}^{\text{w-peak}}(\mathbf{r}^{\star})\le \sum_{i=1}^M\left[\frac{w_ie^{x^{\star}\frac{t_s}{\mathbb{E}[T]}}\left(1+\frac{1}{x^{\star}}\right)}{\min\{b_i,\beta^{\star}\sqrt{w_i}\}}+w_i\right],
\end{align}
where $x^{\star}$, $\beta^{\star}$ are defined in \eqref{x*}, \eqref{condition_beta}. 
\end{lemma}
\begin{proof}
See our technical report \cite{bedewy2019optimal}.
\end{proof}

\textbf{Step 2}: 
We now construct a lower bound on the optimal
value of Problem \textbf{1}. Suppose that $\mathbf{r}=(r_1,\ldots,r_M)$ is a feasible solution to Problem \textbf{1}, such that $r_l >0$ and
\begin{align}
\frac{[1-e^{-r_l\frac{t_s}{\mathbb{E}[T]}}]\sum_{i=1}^Mr_i+r_le^{-r_l\frac{t_s}{\mathbb{E}[T]}}}{\sum_{i=1}^Mr_i+1}\leq b_l, \forall l.
\end{align} 
Because $[1-e^{-r_l(t_s/\mathbb{E}[T])}]\sum_{i=1}^Mr_i+r_le^{-r_l(t_s/\mathbb{E}[T])} > r_l$ for all $l$, $\mathbf{r}$ satisfies $r_l/(\sum_{i=1}^Mr_i+1)\leq b_l$. Hence, the following Problem \textbf{2} has a larger feasible set than Problem \textbf{1}: (Problem \textbf{2})
\begin{align}
\begin{split}\label{problem2}
\!\!\!\!\bar{\Delta}^{\text{w-peak}}_{\text{opt},2}\triangleq\min_{r_l>0}& \sum_{l=1}^M\frac{w_le^{-r_l\frac{t_s}{\mathbb{E}[T]}}}{r_l}e^{\sum_{i=1}^Mr_i\frac{t_s}{\mathbb{E}[T]}}\left(1+\sum_{i=1}^Mr_i\right)+\sum_{l=1}^Mw_l
\end{split}\\
\begin{split}\label{constrain_final_1}
\textbf{s.t.}~&r_l\leq b_l\left(\sum_{i=1}^Mr_i+1\right), \forall l,
\end{split}
\end{align}
where $\bar{\Delta}^{\text{w-peak}}_{\text{opt},2}$ is the optimal value of Problem \textbf{2}. The optimal objective value of Problem \textbf{2} is a lower bound of that of Problem \textbf{1}. We note that the constraint set corresponding to Problem~\textbf{2} is convex. Thus, this relaxation converts the constraint set of Problem~\textbf{1} to a convex one, and hence enables us to obtain a lower bound for the optimal value of Problem \textbf{1}, which is expressed as follows: 
\begin{lemma}\label{lemma_bounds}
If $\sum_{i=1}^Mb_i\ge 1$, then  
\begin{align}
\bar{\Delta}^{\text{w-peak}}_{\text{opt}}\ge\bar{\Delta}^{\text{w-peak}}_{\text{opt},2}\ge\sum_{i=1}^M\left[\frac{w_i}{\min\{b_i,\beta^{\star}\sqrt{w_i}\}}+w_i\right],\label{lower_bound_sum_bi_ge_1}
\end{align}
where $\beta^{\star}$ is the root of \eqref{condition_beta}. 
\end{lemma}
\begin{proof}
See our technical report \cite{bedewy2019optimal}.
\end{proof}

\textbf{Step 3}: After the upper and lower bounds of $\bar{\Delta}^{\text{w-peak}}_{\text{opt}}$ were derived in Steps 1-2, we are ready to analysis their gap. By combining \eqref{upper_bound_sumb_i_ge_1} and \eqref{lower_bound_sum_bi_ge_1}, the sub-optimality gap of the solution $\mathbf{r}^\star$  \eqref{r_f_b_gr_1} - \eqref{condition_beta} is upper bounded by
\begin{align}\label{G1}
\begin{split}
&\left|\bar{\Delta}^{\text{w-peak}}(\mathbf{r}^{\star})\!-\!\bar{\Delta}^{\text{w-peak}}_{\text{opt}}\right|\le \sum_{i=1}^M\frac{w_i\left(e^{x^{\star}\frac{t_s}{\mathbb{E}[T]}}(1\!+\!\frac{1}{x^{\star}})\!-\!1 \right)}{\min\{b_i,\beta^{\star}\sqrt{w_i}\}},
\end{split}
\end{align}
where $x^{\star}$, $\beta^{\star}$ are defined in \eqref{x*}, \eqref{condition_beta}. Next, we characterize the right-hand-side (RHS) of \eqref{G1} by Taylor expansion. For simplicity, let $\epsilon=\frac{t_s}{\mathbb{E}[T]}$. Using the expression for $x^{\star}$ from \eqref{x*}, we have
\begin{equation}\label{x^*t_s}
\begin{split}
x^{\star}\epsilon=&-\frac{\epsilon}{2}+\sqrt{\frac{\epsilon^2}{4}+\epsilon}
=\frac{\epsilon}{\frac{\epsilon}{2}+\sqrt{\frac{\epsilon^2}{4}+\epsilon}}
=\sqrt{\epsilon}+o(\sqrt{\epsilon}).
\end{split}
\end{equation}
Moreover,
\begin{equation}\label{x^*ET}
\begin{split}
x^\star=&-\frac{1}{2}+\sqrt{\frac{1}{4}+\frac{1}{\epsilon}}
=\frac{\frac{1}{\epsilon}}{\frac{1}{2} + \sqrt{\frac{1}{4} + \frac{1}{\epsilon}}}
=\frac{1}{\sqrt{\epsilon}} + o\left(\frac{1}{\sqrt{\epsilon}}\right). 
\end{split}
\end{equation}
 Substituting \eqref{x^*t_s} and~\eqref{x^*ET} in \eqref{G1}, we obtain
\begin{align}
\left|\bar{\Delta}^{\text{w-peak}}(\mathbf{r}^{\star})-\bar{\Delta}^{\text{w-peak}}_{\text{opt}}\right|&
\le \sum_{i=1}^M\frac{w_i[e^{\sqrt{\epsilon} + o(\sqrt{\epsilon})} (1 + \sqrt{\epsilon} + o(\sqrt{\epsilon})) - 1 ]}{\min\{b_i,\beta^{\star}\sqrt{w_i}\}}\nonumber\\&
=\sum_{i=1}^M\frac{w_i[(1\!+\! \sqrt{\epsilon}\! +\!o(\sqrt{\epsilon}))(1 \!+\! \sqrt{\epsilon} \!+\! o(\sqrt{\epsilon})) \!-\! 1 ]}{\min\{b_i,\beta^{\star}\sqrt{w_i}\}}\nonumber\\&
=2 \sqrt{\epsilon} \sum_{i=1}^M\frac{w_i}{\min\{b_i,\beta^{\star}\sqrt{w_i}\}}  + o(\sqrt{\epsilon}),
\end{align}
where the second inequality involves the use of Taylor expansion. This proves Theorem \ref{thm_case_sum_bi_ge_1}. Moreover, we can observe that the gap $\left|\bar{\Delta}^{\text{w-peak}}(\mathbf{r}^{\star})-\bar{\Delta}^{\text{w-peak}}_{\text{opt}}\right|$ in the energy-adequate regime converges to zero at a speed of $O(\sqrt{\epsilon})$, as $\epsilon \to 0$. We also observe that both the upper and lower bounds \eqref{upper_bound_sumb_i_ge_1}, \eqref{lower_bound_sum_bi_ge_1}, converge to $\sum_{i=1}^M[(w_i/\min\{b_i,\beta^{\star}\sqrt{w_i}\})+w_i]$ as $t_s/\mathbb{E}[T]\rightarrow 0$. Thus, this value is the asymptotic optimal value of Problem \textbf{1}. This proves Corollary \ref{cor1}. 
\fi

\ifreport
\subsection{The Proofs of Theorem \ref{thm_case_sum_bi_le_1} and Corollary \ref{cor2}}\label{proof_case_sum_bi_le_1}
Similar to Section \ref{proof_case_sum_bi_ge_1}, we prove Theorem \ref{thm_case_sum_bi_le_1} and
Corollary \ref{cor2} also in  three steps:

\textbf{Step 1}: We show that the proposed solution $\mathbf{r}^{\star}$  \eqref{r_f_b_gr_1} and  \eqref{condition_x*_beta*} - \eqref{feas_fact_eq27} is a feasible solution for Problem \textbf{1}.
\begin{lemma}\label{lemma_feasibility2}
If $\sum_{i=1}^Mb_i< 1$, then the solution $\mathbf{r}^{\star}$  \eqref{r_f_b_gr_1} and  \eqref{condition_x*_beta*} - \eqref{feas_fact_eq27} is feasible for Problem \textbf{1}.
\end{lemma}
\begin{proof}
See Appendix \ref{Appendix_C}.
\end{proof}
Now, we construct an upper bound on the optimal value of Problem \textbf{1} using our proposed solution as follows:
\begin{lemma}\label{lemma_upper_bound_sum_i_b_i_le1}
If $\sum_{i=1}^Mb_i< 1$, then 
\begin{align}\label{upper_bound_sumb_i_le_1}
\begin{split}
\bar{\Delta}^{\text{w-peak}}_{\text{opt}}\le\bar{\Delta}^{\text{w-peak}}(\mathbf{r}^{\star})\le& \sum_{l=1}^M\!\frac{w_l}{b_l}e^{\sum_{i=1}^Mb_ix^{\star}\frac{t_s}{\mathbb{E}[T]}}\left(\frac{1}{x^{\star}}+\sum_{i=1}^Mb_i\!\right)\\&+\sum_{l=1}^Mw_l,
\end{split}
\end{align}
where $x^{\star}$ is defined in \eqref{condition_x*_beta*}. 
\end{lemma}
\begin{proof}
In Lemma \ref{lemma_feasibility2}, we showed that our proposed solution $\mathbf{r}^{\star}$  \eqref{r_f_b_gr_1} and  \eqref{condition_x*_beta*} - \eqref{feas_fact_eq27} is feasible for Problem \textbf{1}. Hence, we substitute this solution into Problem \textbf{1} to obtain the following upper bound:
\begin{align}
\sum_{l=1}^M\!\frac{w_le^{-b_lx^\star\frac{t_s}{\mathbb{E}[T]}}}{b_l}e^{\sum_{i=1}^Mb_ix^{\star}\frac{t_s}{\mathbb{E}[T]}}\left(\frac{1}{x^{\star}}+\sum_{i=1}^Mb_i\!\right)+\sum_{l=1}^Mw_l.
\end{align}
Next, we replace $e^{-b_lx^\star\frac{t_s}{\mathbb{E}[T]}}$ by 1 
to derive another upper bound with a simple expression, which is given by \eqref{upper_bound_sumb_i_le_1}. This completes the proof. 
\end{proof}

\textbf{Step 2}: Similar to the proof in Section \ref{proof_case_sum_bi_ge_1}, we use the relaxed problem, Problem \textbf{2}, to construct a lower bound as follows:
\begin{lemma}\label{lemma_lower_bound_sum_i_b_i_le_1}
If $\sum_{i=1}^Mb_i< 1$, then 
\begin{align}
\bar{\Delta}^{\text{w-peak}}_{\text{opt}}\ge\bar{\Delta}^{\text{w-peak}}_{\text{opt},2}\ge\sum_{l=1}^M\frac{w_l}{b_l}e^{\frac{-\sum_{i=1}^Mb_i}{1-\sum_{i=1}^Mb_i}\frac{t_s}{\mathbb{E}[T]}}+\sum_{l=1}^Mw_l.\label{lower_bound_sum_bi_le_1}
\end{align}
\end{lemma}
\begin{proof}
See Appendix \ref{Appendix_D}.
\end{proof}

\textbf{Step 3}: We now characterize the sub-optimality gap by analyzing the upper and lower bounds constructed above. By combining \eqref{upper_bound_sumb_i_le_1} and \eqref{lower_bound_sum_bi_le_1}, the sub-optimality gap of the solution $\mathbf{r}^\star$  \eqref{r_f_b_gr_1} and  \eqref{condition_x*_beta*} - \eqref{feas_fact_eq27} is upper bounded by
\begin{align}\label{G2}
\begin{split}
&\left|\bar{\Delta}^{\text{w-peak}}(\mathbf{r}^{\star})-\bar{\Delta}^{\text{w-peak}}_{\text{opt}}\right|\\&\le \sum_{l=1}^M\frac{w_l}{b_l}\!\left[ e^{\sum_{i=1}^Mb_ix^{\star}\frac{t_s}{\mathbb{E}[T]}}\left(\frac{1}{x^{\star}}\!+\!\sum_{i=1}^Mb_i\right)\!-\!e^{\frac{-\sum_{i=1}^Mb_i}{1-\sum_{i=1}^Mb_i}\frac{t_s}{\mathbb{E}[T]}}\right].
\end{split}
\end{align}
where $x^{\star}$ is defined in \eqref{condition_x*_beta*}. Next, we characterize the RHS of \eqref{G2} by Taylor expansion. For simplicity, let $\epsilon=t_s/\mathbb{E}[T]$, $Z=(\sum_{i=1}^Mb_i)/(1-\sum_{i=1}^Mb_i)$, and  $k_l=(\sum_{i=1}^Mb_i-b_l)/(1-\sum_{i=1}^Mb_i)^2$. Using Taylor expansion, we are able to obtain the following:
\begin{align}
&\min_l c_l=1+\left(\min_l k_l\right)\epsilon+o(\epsilon),\label{eq37}\\&
\frac{1}{\min_lc_l}=\max_l\frac{1}{c_l}=1+\left(\max_lk_l\right)\epsilon+o(\epsilon).\label{eq38}
\end{align}
Using \eqref{eq37}, \eqref{eq38}, $x^{\star}$ from \eqref{condition_x*_beta*}, and Taylor expansion again, we get
\begin{align}
\begin{split}
e^{\sum_{i=1}^Mb_ix^{\star}\epsilon}&=1+Z\left(1+\left(\min_lk_l\right)\epsilon+o(\epsilon)\right)\epsilon+o(\epsilon)\\&=1+Z\epsilon+o(\epsilon),
\end{split}\label{eq39}\\
\begin{split}
\frac{1}{x^{\star}}+\sum_{i=1}^Mb_i&=\frac{1-\sum_{i=1}^Mb_i}{\min_lc_l}+\sum_{i=1}^Mb_i\\&=1+\left(\max_lk_l\right)\left(1-\sum_{i=1}^Mb_i\right)\epsilon+o(\epsilon),
\end{split}\\
e^{-Z\epsilon}&=1-Z\epsilon+o(\epsilon).\label{eq41}
\end{align}
Substituting \eqref{eq39} - \eqref{eq41} into \eqref{G2}, we get \eqref{sub_opt_gap_eq_2}. This proves Theorem \eqref{thm_case_sum_bi_le_1}. 

Moreover, we observe that the gap $\left|\bar{\Delta}^{\text{w-peak}}(\mathbf{r}^{\star})-\bar{\Delta}^{\text{w-peak}}_{\text{opt}}\right|$ in the energy-scarce regime converges to zero at a speed of $O(\epsilon)$, as $\epsilon \to 0$. Further, both the upper and lower bounds \eqref{upper_bound_sumb_i_le_1}, \eqref{lower_bound_sum_bi_le_1}, converge to $\sum_{i=1}^M[(w_i/b_i)+w_i]$ as $t_s/\mathbb{E}[T]\to 0$. Thus, this value is the asymptotic optimal objective  value of Problem \textbf{1}. This proves Corollary \ref{cor2}.
\fi

%% file: sections/learning_algo.tex
\section{Learning to Optimize Age}\label{sec:learn_algo}
Note that the optimal rate $\br\ust$ in Theorem~\ref{thm_case_sum_bi_ge_1} depends upon the mean transmission time $\bE\left[T\right]$. Since the transmission time also depends upon (possibly) time-varying channel conditions, estimating $\mathbb{E}[T]$ accurately a priori, could be cumbersome. Thus, in this section, we derive learning algorithms that optimize the total weighted average peak age of all sources when the mean transmission time $\bE\left[T\right]$ is unknown to the scheduler. We begin by reducing our system to an equivalent discrete-time Markov chain.

\emph{Contributions and Challenges}: The simplest learning algorithm is called the certainty equivalent rule~\cite{van1981certainty,mania2019certainty,mandl1974estimation,mete2020reward}. In this, the scheduler maintains an empirical estimate of $\bE\left[T\right]$, and utilizes sleep parameters that are optimal when the true value of the mean transmission time is equal to this estimate. The regret of a learning algorithm is the sub-optimality in the performance that results because the algorithm does not know the system parameters. What we are able to show is that by using the CE rule, we are able to get $o(H)$ regret, where $H$ is the time-horizon.
This further implies that the \emph{long-term time-average performance} of our CE algorithm is asymptotically optimal.

This result is important since it is well-known by now~\cite{borkar1979adaptive} that 
in many reinforcement learning problems~\cite{sutton1998reinforcement}, the CE rule fails to be yield long-term time average performance, because it does not yield a correct estimate of the optimal choices. Thus, more complex learning rules, such as optimism in the face of uncertainty~\cite{jaksch2010near,mete2020reward} that utilize confidence balls in addition to the empirical estimates and thus have a significantly higher computational complexity, are required in order to ensure optimality. Our main contribution is to show that the vanilla CE rule yields asymptotically the same long-term time average performance as the scheduler that knows the system parameters in advance, i.e. (with a high probability) the ``sub-optimality gap'' of the CE rule is $o(H)$ where $H$ is the operating time horizon. This means that instead of using more complex learning algorithm such as the UCRL~\cite{jaksch2010near} or RBMLE~\cite{mete2020reward}, one could use CE thereby saving precious computing power and attaining the optimal average performance (asymptotically). We perform a finite-time performance analysis of the CE rule and explicitly quantify its sub-optimality by deriving an upper-bound on its ``regret'', i.e., the gap between its average expected performance, and that resulting from the application of optimal sleep parameter. The problem of designing and analyzing learning algorithms for our setup poses several challenges, primarily because the age process evolves in continuous time on a continuous state-space that is not compact. To address this difficulty, we show that for the purpose of  optimizing average age, we can equivalently work with a discrete-time process. We then utilize several techniques from the theory of general state-space Markov chains~\cite{nummelin2004general} for analyzing the learning regret. 

\emph{Sampling Continuous Time Process}: Consider the multi-source system in which the sleep durations are modulated according to the parameter vector $\br=(r_1,r_2,\ldots,r_M)$. Throughout this section, we let $n\in\bN$ be the discrete time of the sampled system. We sample the original continuous-time system at those time instants when one out of the following events occur:
\begin{itemize}
\item a source $l$ gets channel access and starts transmitting. We say that it wakes up, denoted by $m_{l}(n)=1$,
\item a source $l$ completes packet transmission, and hence goes into sleep mode such that $m_{l}(n)=0$.
\end{itemize}
In what follows, we make this assumption.
\begin{assumption}\label{assum:1}
The transmission times are bounded, i.e., $0\le T\le T_{\max}$ almost surely, where $T_{\max}>0$. Moreover, the probability density function $f(\cdot)$ of $T$ satisfies
\begin{align*}
lb \le f(n) \le ub, \forall y\in [0,T_{\max}],
\end{align*}
where $lb, ub>0$ are upper and lower bounds on the density function.
$\square$
\end{assumption}
Define $s_{l}(n):=(\Delta_{l}(n),m_{l}(n))$, where $\Delta_l(n)$ is the age, and $m_l(n)\in \{0,1\}$ is the mode of user $l$. Define,
\begin{align}\label{def:st}
\bs(n) :=  \left( s_1(n),s_2(n),\ldots,s_M(n)\right). 
\end{align}
As is shown in Lemma~\ref{lemma:equiv_problem} (see Appendix \ref{sec:learn_proof}), for the purpose of adaptively choosing sleep parameters, the process $\bs(n)$ serves as a sufficient-statistics~\cite{striebel1965sufficient} for the optimization problem \eqref{problem1}. In other words, $\bs(n)$ is the state of a Markov decision process. Hence, we will work exclusively with the discrete-time system obtained by sampling the original continuous-time system. We use $\cal{S}$ to denote the state space of a single source, i.e., we have $s_l(n)\in \cal{S}$.  Consider the operation over a time horizon of $H$ discrete time-steps, and let $K_{l}$ denote the (random) number of packets delivered to source $l$ until time $H$. The cumulative cost incurred is given by
\begin{align}\label{def:peak_cost}
C(H) := \sum_{l=1}^{M}  \sum_{i=1}^{K_l} w_l \Delta^{\text{peak} }_{l,i},
\end{align}
where $\Delta^{\text{peak}}_{l,i}$ denotes the $i$-th peak age of source $l$. We let $r_l(n) \in \bR_{+}$ denote the sleep period parameter for source $l$, and denote $\br(n):=\left(r_1(n),r_2(n),\ldots,r_M(n)\right)$. As is shown in Lemma~\ref{lemma:equiv_problem}, the expected value of the cumulative value of peak age can be written as follows,
\begin{align}\label{def:cost}
\bE \left(\sum_{n=1}^{H-1}~g(\bs(n)) \right),
\end{align}
where the function $g$ is described in Lemma~\ref{lemma:equiv_problem}. However, in our setup, the controller that chooses $\br(n)$ does not know the density function $f$ of the packet transmission time, and has to adaptively choose the sleeping period paremeter $\br(n)$ so as to minimize the operating cost~\eqref{def:cost}. 

Let $\cF^{(d)}_t$ denote the sigma-algebra generated by the random variables $\{\bs(i)\}_{i=1}^{n}, \{\br(i)\}_{i=1}^{n-1}$ (the super-script $d$ denotes the fact that we are working with discretized system). A learning policy is a collection of maps $\cF^{(d)}_t\mapsto \br(n), n=1,2,\ldots,H$ that chooses the sleep period parameter adaptively based on  past operation history of the system.~The performance of a learning policy is measured by its regret $R(H)$, which is defined as follows,
\begin{align}\label{def:regret}
R(H) := \sum_{n=1}^{H}  g(\bs(n)) - H \bar{\Delta}^{\text{w-peak}}(\mathbf{r}^{\star}),
\end{align}
where $\bar{\Delta}^{\text{w-peak}}(\mathbf{r}^{\star})$ is the optimal performance when the true system parameter is known and hence the scheduler can implement the optimal rate vector. Throughout this section we use $\theta$ to denote the mean transmission time $\mathbb{E}[T]$. Since the optimal rate depends upon the probability density function $f(\cdot)$ only through its mean $\bE\left[T\right]$, we also denote it by $\br\ust_{\te}$.

\textbf{Certainty Equivalence Learning Algorithm}: We begin with some notations. Let $col(i), i=1,2,\ldots$ be a random variable that is equal to $1$ if there is no collision at time $i$, and is $0$ otherwise.  The empirical estimate of $\theta$ at time $n$ is denoted by $\hat{\theta}(n)$, and given as 
\begin{align}\label{def:point_est}
\hat{\theta}(n) &:=\frac{\sum_{i=1}^{n} T(i) col(i)}{N(n)\vee 1},\\
\mbox{ where }N(n) :&= \sum_{i=1}^{n} col(i),
\end{align}
and $T(i)\in \left[0,T_{\max}\right]$ is the time taken to deliver packet at time $i$.

The learning rule operates in episodes. We let $\tau_k$ be the start time of the $k$-th episode, and let $\cE_k := \left\{\tau_k,\tau_k +1,\ldots,\tau_{k+1}-1\right\}$ be the time-slots that comprise the $k$-th episode, so that the duration of $\cE_k$ is $\tau_{k+1} - \tau_k$ time-slots. We let $\tau_k=2^k$, use $k(n)$ to denote the index of the current episode  at time $n$, and $\te(n)$ to denote the empirical estimate at the beginning of the current episode,  defined by
\begin{align}
k(n) :&= \max\left\{ k: \tau_k \le t  \right\},\label{def:ucb_problem}\\
\te(n) :&= \hat{\theta}(\tau_{k(n)}). \label{def:ucb_problem_1}
\end{align}
Within each single episode the algorithm  implements a single stationary controller that makes decisions only on the basis of the state $\bs(n)$ and the estimate $\te(\tau_k)$ obtained at the beginning of the current ongoing episode $k(n)$. It chooses the sleep period parameter as $\br(n)= \br\ust_{\te(n)}, \forall n\in \cE_k$, i.e., it utilizes the rate vector that is optimal for the system whose mean transmission time is equal to $\theta(n)$. Thus, $\br(n)= \br^{\star}_{\theta(\tau_k)}$ for $\tau_k \le n \le \tau_{k+1}-1$. 

We summarize our learning rule in Algorithm~\ref{algo:ucb}. 
\begin{algorithm}
 \begin{algorithmic}[1]
 \renewcommand{\algorithmicrequire}{\textbf{Input:}}
 \REQUIRE $N,\gamma\ge 4$\\
 Set $\hat{\theta}(1)=.5$.
  \FOR {$n=1,2,\ldots$} 
  \IF {$n = \tau_k$}
  \STATE Calculate $\hat{\theta}(n)$ as in~\eqref{def:point_est} and set $\theta(n)$ as in~\eqref{def:ucb_problem}-\eqref{def:ucb_problem_1}. \\
  \ENDIF\\
  Use $\br(n)=\br\ust_{\te(n)}$
  \ENDFOR
   \caption{Certainty Equivalence Learning for Age Optimization}\label{algo:ucb}
 \end{algorithmic} 
\end{algorithm}
We will analyze its performance under the following assumptions. Throughout, for a vector $\mathbf{x}$, we let $\|\mathbf{x}\|$ denote its Euclidean norm, and $\|\mathbf{x}\|_1$ denote its $1$-norm.
\begin{assumption}\label{assum:bounded}
With a high probability, say greater than $1-\delta$, where $\delta>0$ is a small constant, the state value $\bs(\tau_k)$ at the beginning of each episode $k$ belongs to a compact set $\cK:= \left\{ \bx\in \cS^{M}: \|\bx\|_1 \le K_1  \right\}$,  where $\cal{S}$ is the state space of a single source. 
$\square$
\end{assumption}

The above is not a restrictive assumption, since the scheduler can always ensure that towards the end of each episode, each source receives  a sufficient amount of service in order to ensure this condition. We now make a few assumptions regarding the set $\Theta$ of ``allowable parameters''.

\begin{assumption}\label{assum:2}
Recall that $\br^{\star}_{\te}$ is the optimal sleep parameter when the mean transmission time is equal to $\te$. The following two properties hold for the scheduler that uses $\br(n)\equiv \br^{\star}_{\te}, n\in \bN$. \\
(i) The average cost is finite, i.e.
\begin{align}\label{def:stability}
\limsup_{H\to\infty} \frac{1}{H} \sum_{n=1}^{H} \bE_{\br^{\star}_{\te}}\left(g(\bs(n))\right)  \le K_2 <\infty,
\end{align}
\\
(ii) Each user gets channel access with a non-zero probability 
\begin{align}\label{ineq:channel_access}
\inf_{\te \in \Theta,l\in [M]} \bP\left( ca_{l}(n)=1| \br(n)= \br^{\star}_{\te} \right) & >0,
\end{align}
where $ca_{l}(i)$ is a random variable that is $1$ if source $l$ gets channel access at time $i$, while is $0$ otherwise. We denote
\begin{align}\label{def:cmin}
p_{\min} := \inf_{\te \in \Theta,l\in [M]} \bP\left( ca_{l}(n)=1| \br(n)= \br^{\star}_{\te} \right). 
\end{align}
$\square$
\end{assumption}

It is easily verified that~\eqref{def:stability},~\eqref{ineq:channel_access} hold true whenever the rate vector $\br$ is bounded.

The following result quantifies the learning regret of Algorithm~\ref{algo:ucb}. 
\begin{theorem}\label{th:learn_regret}
Consider the problem of designing a learning algorithm that does not know the statistics of the transmission time $T$, and adaptively chooses the sleep period parameters $\br(n)$ in order to minimize the cumulative peak age of $M$ sources. Let $\delta_1 \in \left(0, p_{\min}\right)$ be a constant. Then, under Assumptions \ref{assum:1}-\ref{assum:2}, the regret of Algorithm~\ref{algo:ucb} can be bounded as follows,
\begin{align*}
\bE\left[R(H)\right] \le & K_2 \max\left\{\frac{\gamma\log H}{( p_{\min} - \sqrt{\delta_1})  \delta^{2}},O\left(\frac{1}{\delta_1}\log H\right)  \right\}+\\& K_2 \frac{\pi^2}{6} + L\sqrt{\frac{H\gamma(\log H)^2}{ ( p_{\min} - \sqrt{\delta_1}) }  }.
\end{align*} 
where $H$ is the operating time horizon, $\gamma\ge4$ is a constant, $K_2,p_{\min}$ are as in Assumption~\ref{assum:2}, and the parameters $\delta,L>0$ are as in Lemma~\ref{lemma:convergence}.
\end{theorem}
\begin{proof}
See Appendix \ref{sec:learn_proof}.
\end{proof}

%% file: sections/Numerical_results.tex
\section{Numerical and Simulation Results}\label{Simulations}
We use Matlab and NS-3 to evaluate the performance of our algorithm.
We use ``age-optimal scheduler'' to denote the sleep-wake scheduler with the sleep period paramters $r_l^{\star}$'s as in \eqref{r_f_b_gr_1}, which was shown to be near-optimal in Theorem \ref{thm_case_sum_bi_ge_1} and Theorem \ref{thm_case_sum_bi_le_1}. By ``throughput-optimal scheduler'', we  refer to the sleep-wake algorithm of~\cite{chen2013life} that is known to achieve the optimal trade-off between the throughput and energy consumption reduction. Moreover, we use ``fixed sleep-rate scheduler'' to denote the sleep-wake scheduler in which the sleep period parameters $r_l$'s are equal for all the sources, i.e., $r_l=k$ for all $l$,  where the parameter $k$ has been chosen so as to satisfy the energy constraints of Problem \textbf{1}.
We also let $\bar{\Delta}_{\text{un}}^{\text{w-peak}}(\mathbf{r})$ denote the unnormalized total weighted average peak age in \eqref{t_avg_peak_age}. Finally, we would like to mention that we do not compare the performance of our proposed algorithm with the CSMA algorithms of~\cite{maatouk2019minimizing,wang2019broadcast} where the goal  was solely to minimize the age. Since they do not incorporate energy constraints, it is not fair to compare the performance of our algorithm with them. 

Unless stated otherwise, our set up is as follows: The average transmission time is $\mathbb{E}[T]=5$ ms. The weights $w_l$'s attached to different sources are generated by sampling from a uniform distribution in the interval $[0, 10]$. The target power efficiencies $b_l$'s are randomly generated according to a uniform distribution in the range $[0, 1]$. 

\subsection{Numerical Evaluations}

\begin{figure}[t]
\centering
\includegraphics[scale=0.4]{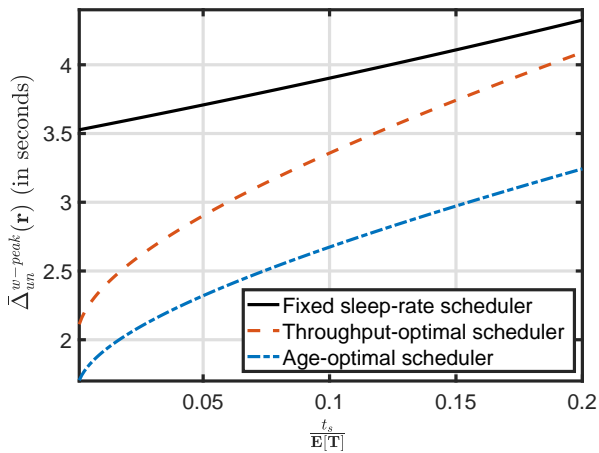}
\centering
\captionsetup{justification=justified}
\caption{Total weighted average peak age $\bar{\Delta}_{\text{un}}^{\text{w-peak}}(\mathbf{r})$ in \eqref{t_avg_peak_age} versus the ratio $\frac{t_s}{\mathbb{E}[T]}$ for $M=10$ sources.}
\label{peak_age_comp}
\end{figure}
Figure~\ref{peak_age_comp} plots the total weighted average peak age $\bar{\Delta}_{\text{un}}^{\text{w-peak}}(\mathbf{r})$ in  \eqref{t_avg_peak_age} as a function of the ratio $\frac{t_s}{\mathbb{E}[T]}$,  where the number of sources is $M=10$. The age-optimal scheduler is seen to outperform the throughput-optimal and Fixed sleep-rate schedulers. This implies that what minimizes the throughput does not necessarily minimize AoI and vice versa. 
 Moreover, we observe that the total weighted average peak age of all schedulers increases as the sensing time increases. This is expected since an increase in the sensing time leads to an increase in the probability of packet collisions, which in turn deteriorates the age performance of these schedulers.

\begin{figure}[t]
\centering
\includegraphics[scale=0.4]{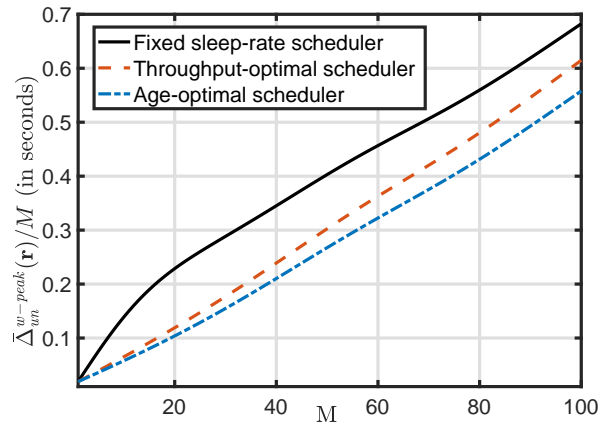}
\centering
\captionsetup{justification=justified}
\caption{Total weighted average peak age $\bar{\Delta}_{\text{un}}^{\text{w-peak}}(\mathbf{r})$ in \eqref{t_avg_peak_age} versus the number of sources $M$, where $\bar{\Delta}_{\text{un}}^{\text{w-peak}}(\mathbf{r})$ has been normalized by $M$ while plotting.}
\label{peak_age_comp_num}
\end{figure}
We then scale the number of sources $M$, and plot $\bar{\Delta}_{\text{un}}^{\text{w-peak}}(\mathbf{r})$  in  \eqref{t_avg_peak_age} as a function of $M$ in Figure~\ref{peak_age_comp_num}. While plotting, we normalize the performance by the number of sources $M$. The sensing time $t_s$ is fixed at $t_s=40~\mu$s. The weights $w_l$'s corresponding to different sources are randomly generated uniformly  within the range $[0, 2]$. The age-optimal scheduler is shown to outperform other schedulers uniformly for all values of $M$. Moreover, as we can observe, the average peak age of the sources under age-optimal scheduler increases up to around 0.55 seconds only, while the number of sources rises from 1 to 100. This indicates the  robustness of our algorithm to changes in the number of sources in a network.

\begin{figure}[t]
\centering
\includegraphics[scale=0.4]{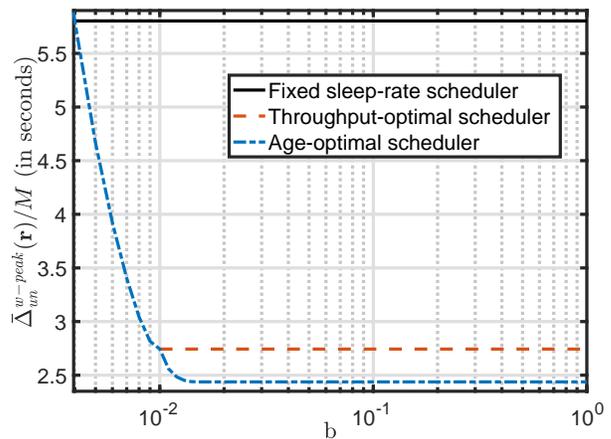}
\centering
\captionsetup{justification=justified}
\caption{Total weighted average peak age $\bar{\Delta}_{\text{un}}^{\text{w-peak}}(\mathbf{r})$ in \eqref{t_avg_peak_age} versus the target power efficiency $b$ for  $M=100$ sources, where $\bar{\Delta}_{\text{un}}^{\text{w-peak}}(\mathbf{r})$ has been normalized by $M$ while plotting.}
\label{peak_age_comp_lower_bound}
\end{figure}
In Figure \ref{peak_age_comp_lower_bound}, we fix the value of $M$ as $100$ and the target power efficiencies at the same value for all the sources, i.e., $b_l= b$ for all $l$. We then vary the parameter $b$ and plot the resulting performance. While plotting, we normalize the performance by the number of sources $M$.  We exclude the simulation of the throughput-optimal scheduler for $b< 0.01$ since the sleeping period parameters that are proposed in \cite{chen2013life} are not feasible for Problem \textbf{1} in the energy-scarce regime, i.e., when $\sum_{i=1}^Mb_i<1$. The age-optimal scheduler outperforms the other  schedulers. Moreover, its performance is a decreasing function of $b$, and then settles at a constant value. This occurs because our proposed solution in~\eqref{r_f_b_gr_1}  is a function solely of the weights $w_l$'s and $\beta^\star$ when $b$ exceeds some value. Thus, the performance of the proposed scheduler saturates after this value of $b$.

\begin{figure}[t]
\centering
\includegraphics[scale=0.4]{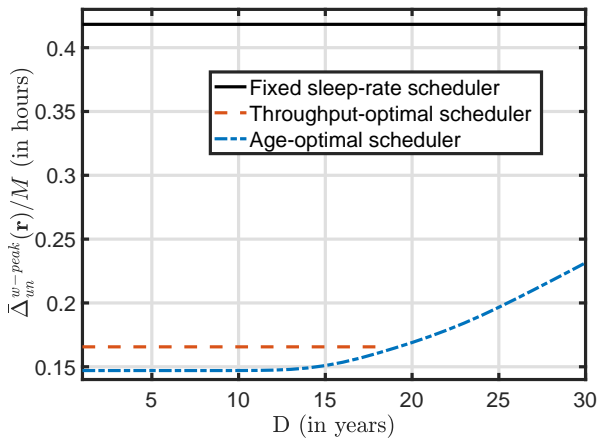}
\centering
\captionsetup{justification=justified}
\caption{Total weighted average peak age $\bar{\Delta}_{\text{un}}^{\text{w-peak}}(\mathbf{r})$ in \eqref{t_avg_peak_age} versus the target lifetime $D$ for a dense network with  $M=10^5$ sources, where $\bar{\Delta}_{\text{un}}^{\text{w-peak}}(\mathbf{r})$ has been normalized by $M$ while plotting. Since the throughput--optimal scheduler is infeasible for values of $D$ greater than $18$ years, we do not plot its performance for these values.}
\label{peak_age_comp_dense_network}
\end{figure}
We now show the effectiveness of the proposed scheduler when deployed in ``dense networks''~\cite{kowshik2019energy,kowshik2019fundamental}. Dense networks are characterized by a large number of sources connected to a single AP. We fix $M$ at $10^5$ sources, and take the target lifetimes of the sources to be equal, i.e., $D_l= D$ for all $l$. The weights $w_l$'s corresponding to different sources are generated randomly by sampling from the uniform distribution in the range $[0, 2]$. We let the initial battery level $B_l= 8$ mAh for all $l$ and the output voltage is 5 Volt. We also let the energy consumption in a transmission mode to be 24.75 mW for all sources. We vary the parameter $D$ and plot the resulting performance in Figure \ref{peak_age_comp_dense_network}. While plotting, we normalize the performance by the number of sources $M$. We exclude simulations for the throughput-optimal scheduler for values of $D$ for which the scheduler is infeasible, i.e., its cumulative energy consumption exceeds the total allowable energy consumption. The age-optimal scheduler is seen to outperform the others. 
As observed in Figure~\ref{peak_age_comp_dense_network}, under the age-optimal scheduler, sources can be active for up to 25 years, while simultaneously achieving a decent average peak age of around .2 hour, i.e., 12 minutes. This makes it suitable for dense networks, where it is crucial that the sources are necessarily active for many years.
\subsection{NS-3 Simulation}\label{ns3sim}
We use NS-3 \cite{ns_3_2_0} to investigate the effect of our model  assumptions on the performance of  age-optimal scheduler in a more practical situation. We simulate the age-optimal scheduler by using IEEE 802.11b while disabling the RTS-CTS and modifying the back-off times to be exponentially distributed in the MAC layer. Our simulation results are averaged over 5 system realizations. The UDP saturation conditions are satisfied such that the source nodes always have packets to send.

 Our simulation consists of a WiFi network with 1 AP and 3 associated source nodes in a field of size 50m $\times$ 50m. We set the sensing threshold to -100 dBm which covers a range of 110m. Thus, all sources can hear each other. The initial battery level of each source is 60 mAh, where the output voltage is 5 Volt. For each source, the power consumption in the transmission mode is 24.75 mW, and the power consumption in the sleep mode is 15 $\mu$W. Moreover, all weights are set to unity, i.e., $w_l=1$ for all $l$.

\begin{figure}[t]
\centering
\includegraphics[scale=0.4]{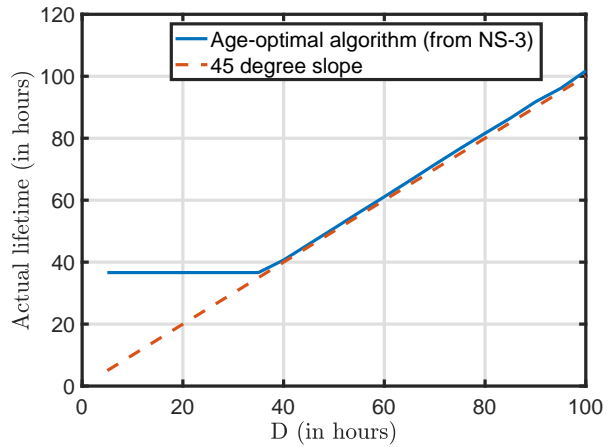}
\centering
\captionsetup{justification=justified}
\caption{The average actual lifetime versus the target lifetime $D$.}
\label{ns3_actual_lifetime}
\end{figure}
Figure \ref{ns3_actual_lifetime} plots the average actual lifetime of the sources versus the target lifetime, where we take the target lifetimes of all sources to be equal, i.e., $D_l=D$ for all $l$. As we can observe, the actual lifetime of the age-optimal scheduler always achieves the target lifetime. This suggests that our assumptions (i.e.,  (i) omitting the power dissipation in the sleep mode and in the sensing times, (ii) the average transmission times and collision times  are equal to each other) do not affect the performance of the algorithm which reaches its target lifetime.

\begin{figure}[t]
\centering
\includegraphics[scale=0.4]{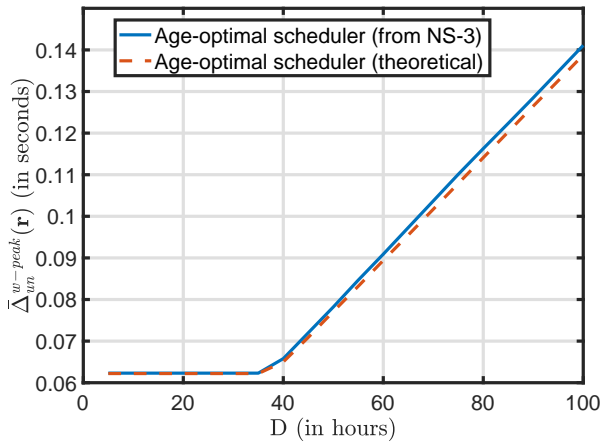}
\centering
\captionsetup{,justification=justified}
\caption{Total weighted average peak age $\bar{\Delta}_{\text{un}}^{\text{w-peak}}(\mathbf{r})$ versus the target lifetime $D$.}
\label{ns3_total_avg_peak}
\end{figure}
Figure \ref{ns3_total_avg_peak} plots the total weighted average peak age versus the target lifetime, where again we take the target lifetimes of all sources to be equal, i.e., $D_l=D$ for all $l$. The age-optimal scheduler (theoretical) curve is obtained using  \eqref{t_avg_peak_age}, while
the age-optimal scheduler (from NS-3) curve is obtained using the NS-3 simulator.  As we can observe, the difference between the plotted curves does not exceed 2\% of the age-optimal scheduler (theoretical) performance.
This emphasizes the negligible impact of our assumptions on the performance of our proposed algorithm.

%% file: sections/conclusion.tex
\section{Conclusions}\label{Concl}
We designed an efficient sleep-wake scheduling algorithm for wireless networks that attains the optimal trade-off between minimizing the AoI and energy consumption. Since the associated optimization problem is non-convex, in general we could not hope to solve it for all values of the system parameters. However, in the regime when the carrier sensing time $t_s$ is negligible as compared to the average transmission time $\mathbb{E}[T]$, we were able to provide a near-optimal solution. Moreover, the proposed solution is in a simple form that allowed us to design an easy-to-implement algorithm to obtain the solution. Furthermore, we showed that the performance of our proposed algorithm is asymptotically no worse than that of the optimal synchronized scheduler, as $t_s/\mathbb{E}[T]\to 0$. Finally, when the mean transmission time is unknown, we devise a reinforcement learning algorithm that adaptively learns the unknown parameter. 


\section{Acknowledgements}
The authors appreciate Jiayu Pan and Shaoyi Li for their great efforts in obtaining the ns-3 simulation results. 

%% file: sections/appendices_sec.tex
\ifreport 
\section{Appendix}
\appendices
\section{Derivation of (\ref{access_prob_in_agiven_cycle})}\label{Appendix_A'}
Define $S_l$ as the residual sleeping period of source~$l$ after a sleep-wake cycle is over. Due to the memoryless property of exponential distribution, since the sleeping period of source~$l$ is exponentially distributed with mean value $\mathbb{E}[T]/r_l$, $S_l$ is also exponentially distributed with mean value $\mathbb{E}[T]/r_l$. According to the proposed sleep-wake scheduler, source~$l$ gains access to the channel and transmits successfully in a given cycle if $S_i\ge S_l+t_s$ for all $i\neq l$. Hence, we have
\begin{align}
\alpha_l&=\mathbb{P}(S_i\ge S_l+t_s,~\forall i\neq l)\\&
\stackrel{(a)}{=}\mathbb{E}[\mathbb{P}(S_i\ge S_l+t_s,~\forall i\neq l\vert S_l)]\\&
\stackrel{(b)}{=}\mathbb{E}\left[\prod_{i\neq l}\mathbb{P}(S_i\ge S_l+t_s\vert S_l)\right]\\&
=\int_0^\infty\left[\prod_{i\neq l}e^{-r_i\frac{s_l+t_s}{\mathbb{E}[T]}}\right]\frac{r_l}{\mathbb{E}[T]}e^{-r_l\frac{s_l}{\mathbb{E}[T]}}ds_l\\&
=\frac{r_l e^{r_l\frac{t_s}{\mathbb{E}[T]}}}{e^{\sum_{i=1}^Mr_i\frac{t_s}{\mathbb{E}[T]}}\sum_{i=1}^Mr_i},
\end{align}
where (a) is due to $\mathbb{P}[A]=\mathbb{E}[\mathbb{P}(A\vert B)]$, and (b) is due to the fact that $S_l$ is independent for different sources.\qed

\section{Derivation of (\ref{sigma_l})}\label{Appendix_A''}
Recall the definition of $S_l$ at the beginning of Appendix \ref{Appendix_A'}. Moreover, define $P_l$ as the probability that source~$l$ transmits a packet in a given cycle, regardless whether packet collision occurs or not. For the sleep-wake scheduling mechanism the we are utilizing here, source~$l$ transmits in a given cycle as long as no other source wakes up before $S_l-t_s$, i.e., $S_i\ge S_l-t_s$ for all $i\neq l$. Hence, we have
\begin{align}
P_l&=\mathbb{P}(S_i\ge S_l-t_s,~\forall i\neq l)\\&
=\mathbb{P}(S_i\ge S_l-t_s,~\forall i\neq l,~S_l\ge t_s)+\mathbb{P}(S_l<t_s),\label{P_l}
\end{align}
where the first term in the RHS is given by
\begin{align}
&\mathbb{P}(S_i\ge S_l-t_s\ge 0,~\forall i\neq l)\\
=&\mathbb{E}[\mathbb{P}(S_i\ge S_l-t_s\ge 0,~\forall i\neq l\vert S_l)]\\
=&\mathbb{E}\left[\prod_{i\neq l}\mathbb{P}(S_i\ge S_l-t_s\ge 0\vert S_l)]\right]\\
=&\int_{t_s}^\infty\left[\prod_{i\neq l}e^{-r_i\frac{s_l-t_s}{\mathbb{E}[T]}}\right]\frac{r_l}{\mathbb{E}[T]}e^{-r_l\frac{s_l}{\mathbb{E}[T]}}ds_l\\
=&e^{-r_l\frac{t_s}{\mathbb{E}[T]}}\frac{r_l}{\sum_{i=1}^Mr_i}.\label{App15_56}
\end{align}
Since $S_l$ is exponentially distributed with mean value $\mathbb{E}[T]/r_l$, we can determine the second term in the RHS of \eqref{P_l} as follows:
\begin{align}\label{eq69_2}
\mathbb{P}(S_l<t_s)=1-e^{-r_l\frac{t_s}{\mathbb{E}[T]}}.
\end{align}
Substituting \eqref{App15_56} and \eqref{eq69_2} back into \eqref{P_l}, we get
\begin{align}
P_l=1-e^{-r_l\frac{t_s}{\mathbb{E}[T]}}+e^{-r_l\frac{t_s}{\mathbb{E}[T]}}\frac{r_l}{\sum_{i=1}^Mr_i}.
\end{align}
Let $\alpha_{\text{col}}$ denote the collision probability in a given cycle. We have $\alpha_{\text{col}}=1-\sum_{i=1}^M\alpha_i$, because each cycle includes either a successful transmission or a collision. Moreover, let $\mathbb{E}[\mathbf{Idle}]$ denote the mean of the idle duration in a cycle. By the renewal theory in stochastic processes \cite{gallager2012discrete}, $\sigma_l$ is given by
\begin{align}
\sigma_l&=\frac{P_l\mathbb{E}[T]}{(\sum_{i=1}^M\alpha_i+\alpha_{\text{col}})\mathbb{E}[T]+\mathbb{E}[\mathbf{Idle}]}\\&
=\frac{P_l\mathbb{E}[T]}{\mathbb{E}[T]+\frac{\mathbb{E}[T]}{\sum_{i=1}^Mr_i}}\\&
=\frac{[1-e^{-r_l\frac{t_s}{\mathbb{E}[T]}}]\sum_{i=1}^Mr_i+r_le^{-r_l\frac{t_s}{\mathbb{E}[T]}}}{\sum_{i=1}^Mr_i+1}.
\end{align}
\qed

\section{Proof of Lemma \ref{lemma_feasibility}}\label{Appendix_A}
First of all, we need to show that \eqref{condition_beta} has a solution for $\beta^{\star}$.
\begin{lemma}\label{lemma_solution_19}
Suppose that $w_l>0$, and $b_l>0$ for all $l$. If $\sum_{i=1}^M b_i \geq 1$, then \eqref{condition_beta} has a unique solution on $[0,\max_l(b_l/\sqrt{w_l})]$; otherwise, \eqref{condition_beta} has no solution.  
\end{lemma}
\begin{proof}
It is clear that if $\sum_{i=1}^M b_i =1$, then $\beta^\star$ satisfies \eqref{condition_beta} if and only if $\beta^\star\ge \max_l(b_l/\sqrt{w_l})$. Hence, \eqref{condition_beta} has a unique solution on $[0,\max_l(b_l/\sqrt{w_l})]$ in this case. We now focus on the case of $\sum_{i=1}^M b_i >1$. In this case, we have the following: 
\begin{itemize}
\item If $\beta^\star = 0$, then $\sum_{i=1}^M\min\{b_i,\beta^{\star}\sqrt{w_i}\}=0$.
\item If $\beta^\star =\max_l(b_l/\sqrt{w_l})$, then $\sum_{i=1}^M\min\{b_i,\beta^{\star}\sqrt{w_i}\}>1$. 
\item The left hand side (LHS) of \eqref{condition_beta} is strictly increasing and continuous in $\beta^\star$ on $[0,\max_l(b_l/\sqrt{w_l})]$.
\end{itemize}
As a result, \eqref{condition_beta} has a unique solution on $[0,\max_l(b_l/\sqrt{w_l})]$ in this case as well. Finally, if $\sum_{i=1}^M b_i <1$, then $\sum_{i=1}^M\min\{b_i,\beta^{\star}\sqrt{w_i}\}\le \sum_{i=1}^M b_i < 1$. Hence, \eqref{condition_beta} has no solution if $\sum_{i=1}^M b_i <1$. This completes the proof. 
\end{proof}
Since we have $\sum_{i=1}^M b_i \geq 1$, Lemma \ref{lemma_solution_19} implies that \eqref{condition_beta} has a solution for $\beta^\star$. Now, we are ready to prove Lemma \ref{lemma_feasibility}. Consider the following constraints:
\begin{align}\label{equivelent_constrant}
\frac{r_l\frac{t_s}{\mathbb{E}[T]}\sum_{i=1}^Mr_i+r_l}{\sum_{i=1}^Mr_i+1}\leq b_l, \forall l.
\end{align}
Since we have 
\begin{align}
&1-e^{-r_l\frac{t_s}{\mathbb{E}[T]}}\le r_l\frac{t_s}{\mathbb{E}[T]},\\&
e^{-r_l\frac{t_s}{\mathbb{E}[T]}}\leq 1,
\end{align}
then, 
\begin{align}
[1-e^{-r_l\frac{t_s}{\mathbb{E}[T]}}]\sum_{i=1}^Mr_i+r_le^{-r_l\frac{t_s}{\mathbb{E}[T]}}\leq r_l\frac{t_s}{\mathbb{E}[T]}\sum_{i=1}^Mr_i+r_l.
\end{align}
Thus, if the constraints in \eqref{equivelent_constrant} are satisfied for a given solution $\mathbf{r}$, then the constraints of Problem \textbf{1} are satisfied as well. We can observe that the constraints in \eqref{equivelent_constrant} are equivalent to the following set of constraints:
\begin{align}\label{inequality_feasible}
\begin{split}
&r_l\le b_l\frac{x+1}{1+\frac{t_s}{\mathbb{E}[T]}x},\forall l\\ 
&\sum_{i=1}^Mr_i=x.
\end{split}
\end{align}
Now, it is easy to show that if $x\le \sqrt{\mathbb{E}[T]/t_s}$, then  $x\le (x+1)/[1+(t_s/\mathbb{E}[T])x]$. Meanwhile, our proposed solution $\mathbf{r}^{\star}$  \eqref{r_f_b_gr_1} - \eqref{condition_beta} satisfies $\sum_{i=1}^Mr_i^{\star}=x^{\star}$. Thus, if we can show that $x^{\star}\le \sqrt{\mathbb{E}[T]/t_s}$, then
\begin{align}
r_l^{\star}=\min\{b_l,\beta^{\star}\sqrt{w_l}\}x^{\star}\le b_lx^{\star}\le b_l \frac{x^{\star}+1}{1+\frac{t_s}{\mathbb{E}[T]}x^{\star}},
\end{align}
and the constraints in \eqref{inequality_feasible} hold for our proposed solution $\mathbf{r}^\star$. What remains is to prove that $x^{\star}\le \sqrt{\mathbb{E}[T]/t_s}$. We have
\begin{align}
x^{\star}=&\frac{-1}{2}+\sqrt{\frac{1}{4}+\frac{\mathbb{E}[T]}{t_s}}\\
=& \frac{\frac{\mathbb{E}[T]}{t_s}}{\frac{1}{2}+\sqrt{\frac{1}{4}+\frac{\mathbb{E}[T]}{t_s}}}\\
\le& \frac{\frac{\mathbb{E}[T]}{t_s}}{\sqrt{\frac{\mathbb{E}[T]}{t_s}}}= \sqrt{\frac{\mathbb{E}[T]}{t_s}}.
\end{align}
Hence, our proposed solution $\mathbf{r}^{\star}$  \eqref{r_f_b_gr_1} - \eqref{condition_beta} satisfies  \eqref{inequality_feasible}, which implies \eqref{equivelent_constrant}. This completes the proof. \qed

\section{Proof of Lemma \ref{lemma_bounds}}\label{Appendix_B}
By replacing $e^{-r_l(t_s/\mathbb{E}[T])}e^{\sum_{i=1}^Mr_i(t_s/\mathbb{E}[T])}$ in \eqref{problem2} of Problem \textbf{2} by 1, we obtain the following optimization problem:
\begin{align}
\begin{split}\label{lower_bound_optimization_problme1}
\min_{r_l>0}& \sum_{l=1}^M\frac{w_l}{r_l}\left(1+\sum_{i=1}^Mr_i\right)+\sum_{l=1}^Mw_l
\end{split}\\
\textbf{s.t.}~&r_l\leq b_l\left(\sum_{i=1}^Mr_i+1\right), \forall l.
\end{align}
Since $e^{-r_l(t_s/\mathbb{E}[T])}e^{\sum_{i=1}^Mr_i(t_s/\mathbb{E}[T])}\geq 1$, Problem \eqref{lower_bound_optimization_problme1} serves as a lower bound of Problem \textbf{2}, and hence a lower bound of Problem \textbf{1} as well. Define an auxiliary variable $y=\sum_{i=1}^Mr_i+1$. By this, we solve a two-layer nested optimization problem. In the inner layer, we optimize $\mathbf{r}$ for a given $y$. After solving $\mathbf{r}$, we will optimize $y$ in the outer layer. Now, fix the value of $y$, we obtain the following optimization problem (the inner layer):
\begin{align}
\begin{split}\label{optimization_auxiliary2}
\min_{r_i>0}& \sum_{i=1}^M\left[\frac{w_iy}{r_i}+w_i\right]
\end{split}\\
\textbf{s.t.}~&r_l\leq b_ly, \forall l,\label{const12}\\&
\sum_{i=1}^Mr_i+1=y.\label{const22}
\end{align}
The objective function in \eqref{optimization_auxiliary2} is a convex function. Moreover, the constraints in \eqref{const12} and \eqref{const22} are affine. Hence, Problem \eqref{optimization_auxiliary2} is convex. We use the Lagrangian duality approach to solve Problem  \eqref{optimization_auxiliary2}. Problem \eqref{optimization_auxiliary2} satisfies Slater's conditions. Thus, the Karush-Kuhn-Tucker (KKT) conditions are  both  necessary  and  sufficient for optimality \cite{boyd2004convex}. Let $\mathbf{\gamma}=(\gamma_1, \ldots, \gamma_M)$ and $\mu$ be the Lagrange multipliers associated with constraints \eqref{const12} and \eqref{const22}, respectively. Then, the Lagrangian of Problem \eqref{optimization_auxiliary2} is given by
\begin{equation}\label{lag88}
\begin{split}
\!\!\!\!L(\mathbf{r},\mathbf{\gamma},\mu)=&\sum_{i=1}^M\left[\frac{w_iy}{r_i}+w_i\right]\\&+\sum_{i=1}^M\!\gamma_i(r_i\!-\!b_iy)+\mu\left(\sum_{i=1}^Mr_i\!+\!1\!-\!y\right).\!\!\!\!
\end{split}
\end{equation}
Take the derivative of \eqref{lag88} with respect to $r_l$ and set it equal to 0, we get
\begin{align}
\frac{-w_ly}{r_l^2}+\gamma_l+\mu=0.
\end{align}
This and KKT conditions imply
\begin{align}
&r_l=\sqrt{\frac{w_ly}{\gamma_l+\mu}},\\
&\gamma_l\geq 0, r_l-b_ly\leq 0,\\
&\gamma_l( r_l-b_ly)=0,\\
&\sum_{i=1}^Mr_i+1=y.
\end{align}
If $\gamma_l=0$, then $r_l=\sqrt{(w_ly)/\mu}$ and $r_l\leq b_ly$; otherwise, if $\gamma_l>0$, then $r_l=b_ly$ and $r_l<\sqrt{(w_ly)/\mu}$. Hence, we have
\begin{equation}
r_l=\min\left\lbrace b_ly,\sqrt{\frac{w_ly}{\mu^{\star}}}\right\rbrace,
\end{equation}
where by \eqref{const22}, $\mu^{\star}$ satisfies 
\begin{align}
\sum_{i=1}^M\min\left\lbrace b_iy,\sqrt{\frac{w_iy}{\mu^{\star}}}\right\rbrace+1=y.
\end{align}
We can observe that $\mu^{\star}$ is a function of $y$. Because of that, we can define $\beta^{\star}(y)=\sqrt{1/(y\mu^{\star})}$, which is a function of $y$ as well. Then, the optimum solution to \eqref{optimization_auxiliary2} can be rewritten as
\begin{equation}\label{solution_cond_1_2_2}
r_l=\min\{b_l,\beta^{\star}(y)\sqrt{w_l}\}y, \forall l,
\end{equation}
where $\beta^{\star}(y)$ satisfies
\begin{equation}\label{solution_cond_1_1_1}
\sum_{i=1}^M\min\{b_i,\beta^{\star}(y)\sqrt{w_i}\}+\frac{1}{y}=1.
\end{equation}
Substituting \eqref{solution_cond_1_2_2} and \eqref{solution_cond_1_1_1} back in Problem \eqref{optimization_auxiliary2}, we get the following optimization problem (the outer layer):
\begin{align}
\begin{split}\label{eqa55}
\min_{y>1}& \sum_{i=1}^M\left[\frac{w_i}{\min\{b_i,\beta^{\star}(y)\sqrt{w_i}\}}+w_i\right]
\end{split}\\
\textbf{s.t.}~&\sum_{i=1}^M\min\{b_i,\beta^{\star}(y)\sqrt{w_i}\}+\frac{1}{y}=1.\label{eqa56}
\end{align}
Problem \eqref{eqa55} serves as a lower bound of Problem \textbf{2}, and hence a lower bound of Problem \textbf{1}. We can observe that the objective function in \eqref{eqa55} is decreasing in $\beta^{\star}(y)$. Moreover, \eqref{eqa56} implies that $\beta^{\star}(y)$ is strictly increasing in $y$ if $\sum_{i=1}^Mb_i\geq 1$. As a result, $y=\infty$ is the optimal solution of Problem \eqref{eqa55}. At the limit, the constraint \eqref{eqa56} converges to \eqref{condition_beta}. Since $\beta^{\star}$ serves as a solution for \eqref{condition_beta}, we can deduce that $\lim_{y\to\infty}\beta^\star (y)=\beta^\star$. Thus, we have the following lower bound:
\begin{equation}\label{lower_bound_final_sumbi_ge1}
\bar{\Delta}^{\text{w-peak}}_{\text{opt}}\ge \bar{\Delta}^{\text{w-peak}}_{\text{opt},2}\geq \sum_{i=1}^M\left[\frac{w_i}{\min\{b_i,\beta^{\star}\sqrt{w_i}\}}+w_i\right].
\end{equation}
This completes the proof. \qed

\section{Proof of Lemma \ref{lemma_feasibility2}}\label{Appendix_C}
Because $1-e^{-x}\leq x$, we can obtain
\begin{align}
\begin{split}
&r_le^{-r_l\frac{t_s}{\mathbb{E}[T]}}+[1-e^{-r_l\frac{t_s}{\mathbb{E}[T]}}]\sum_{i=1}^Mr_i\\&=r_l+[1-e^{-r_l\frac{t_s}{\mathbb{E}[T]}}]\left(\sum_{i=1}^Mr_i-r_l\right)\\&\leq  r_l+r_l\frac{t_s}{\mathbb{E}[T]}\left(\sum_{i=1}^mr_i-r_l\right),
\end{split}
\end{align}
Hence, if $\mathbf{r}$ satisfies the constraint
\begin{align}\label{tighter_constraint_ac}
\frac{r_l+r_l\frac{t_s}{\mathbb{E}[T]}\left(\sum_{i=1}^Mr_i-r_l\right)}{\sum_{i=1}^Mr_i+1}\le b_l,
\end{align}
then $\mathbf{r}$ also satisfies the constraint of Problem \textbf{1} in \eqref{problem1}. Consider the following set of solution indexed by a parameter $c>0$:
\begin{align}
&r_l=cu_l,~\forall l,\label{proposed_sol_ac}\\
&u_l=\frac{b_l}{1-\sum_{i=1}^Mb_i},~\forall l\label{proposed_sol_ac103}
\end{align}
We want to find a $c$ such that the solution in \eqref{proposed_sol_ac} and \eqref{proposed_sol_ac103} is feasible for Problem \textbf{1}. To achieve this, we first substitute the solution \eqref{proposed_sol_ac} and \eqref{proposed_sol_ac103} into the constraint \eqref{tighter_constraint_ac}, and get
\begin{align}\label{lem4_4eq104}
\frac{cu_l+c^2u_l\frac{t_s}{\mathbb{E}[T]}\left(\sum_{i=1}^Mu_i-u_l\right)}{c\sum_{i=1}^Mu_i+1}\le b_l.
\end{align} 
If equality is satisfied in \eqref{lem4_4eq104}, we can
obtain the following quadratic equation for c: 
\begin{align}\label{lem4_4eq105}
c^2\left[u_l\frac{t_s}{\mathbb{E}[T]}\left(\sum_{i=1}^Mu_i\!-\!u_l\right)\right]\!+\!c\left(u_l\!-\!b_l\sum_{i=1}^Mu_i\right)\!-\!b_l=0.
\end{align}
The solution to \eqref{lem4_4eq105} is given by $c_l$ in \eqref{feasible_factor}. Hence, $r_l = c_l u_l$ is feasible for the constraint \eqref{tighter_constraint_ac} for source~$l$. 

As feasibility for one source only is insufficient, we further prove that the solution in \eqref{proposed_sol_ac} and \eqref{proposed_sol_ac103} with $c=\min_lc_l$ is feasible for satisfying the energy constraints of all sources $l = 1,\ldots, M$. To that end, let us consider the monotonicity of the LHS of \eqref{lem4_4eq104}. By taking the  derivative with respect to $c$, we get
\begin{align}
\frac{u_l\frac{t_s}{\mathbb{E}[T]}\left(\sum_{i=1}^Mu_i-u_l\right)\left(c^2\sum_{i=1}^Mu_i+2c\right)+u_l}{(c\sum_{i=1}^Mu_i+1)^2}>0.
\end{align}
Hence, 
\begin{align}\label{lem4_4eq106}
r_l = \left(\min_{l} c_l\right) u_l, ~\forall l,
\end{align}
is feasible for the energy constraints of all sources $l = 1,\ldots, M$. After some manipulations, the solution in \eqref{proposed_sol_ac103} and \eqref{lem4_4eq106} are equivalently expressed as \eqref{r_f_b_gr_1} and  \eqref{condition_x*_beta*} - \eqref{feas_fact_eq27}. This completes the proof. \qed

\section{Proof of Lemma \ref{lemma_lower_bound_sum_i_b_i_le_1}}\label{Appendix_D}
By replacing $e^{-r_l(t_s/\mathbb{E}[T])}/r_l$ by $e^{-\sum_{i=1}^Mr_i(t_s/\mathbb{E}[T])}/[$ $b_l(\sum_{i=1}^Mr_i+1)]$  and $e^{\sum_{i=1}^Mr_i(t_s/\mathbb{E}[T])}$ by 1 in \eqref{problem2} of Problem \textbf{2}, we obtain the following optimization problem:
\begin{align}
\begin{split}\label{lower_bound_case_bf_le_1_1}
\min_{r_l>0}& \sum_{l=1}^M\frac{w_le^{-\sum_{i=1}^Mr_i\frac{t_s}{\mathbb{E}[T]}}}{b_i}+\sum_{l=1}^Mw_l
\\
\textbf{s.t.}~&r_l\leq b_l\left(\sum_{i=1}^Mr_i+1\right), \forall l.
\end{split}
\end{align}
Since  $r_l\leq b_l(\sum_{i=1}^Mr_i$ $+1)$, we have
\begin{align}
 \frac{e^{-r_l\frac{t_s}{\mathbb{E}[T]}}}{r_l}\geq \frac{e^{-\sum_{i=1}^Mr_i\frac{t_s}{\mathbb{E}[T]}}}{b_l\left(\sum_{i=1}^Mr_i+1\right)}.
 \end{align}
 Moreover, we have  $e^{\sum_{i=1}^Mr_i(t_s/\mathbb{E}[T])}\geq 1$. Thus, Problem \eqref{lower_bound_case_bf_le_1_1} serves as a lower bound of Problem \textbf{2}, and hence a lower bound of Problem \textbf{1} as well. By removing the constant term $\sum_{l=1}^Mw_l$ in the objective function of Problem  \eqref{lower_bound_case_bf_le_1_1} and then taking the logarithm, Problem \eqref{lower_bound_case_bf_le_1_1} is reformulated as 
 \begin{align}
\begin{split}\label{problem64}
\min_{r_i>0}& \log\left(\sum_{i=1}^M\frac{w_i}{b_i}\right)-\sum_{i=1}^Mr_i\frac{t_s}{\mathbb{E}[T]}
\\
\textbf{s.t.}~&r_l\leq b_l\left(\sum_{i=1}^Mr_i+1\right), \forall l.
\end{split}
\end{align}
Obviously, Problem \eqref{problem64} is a convex optimization problem and satisfies Slater's  conditions. Thus,  the KKT conditions are are necessary and sufficient for optimality. Let $\mathbf{\tau}=(\tau_1,\ldots,\tau_M)$ be the Lagrange multipliers associated with the  constraints of Problem  \eqref{problem64}. Then, the Lagrangian of Problem \eqref{problem64} is given by 
\begin{equation}\label{lag110}
\begin{split}
L(\mathbf{r}, \mathbf{\tau})=&\log\left(\sum_{i=1}^M\frac{w_i}{b_i}\right)-\left(\sum_{i=1}^Mr_i\frac{t_s}{\mathbb{E}[T]}\right)\\&+\sum_{i=1}^M\tau_i\left[r_i-b_i\left(\sum_{i=1}^Mr_i+1\right)\right].
\end{split}
\end{equation}
Take the derivative of \eqref{lag110} with respect to $r_l$ and set it equal to 0, we get
\begin{align}
\frac{-t_s}{\mathbb{E}[T]}+\tau_l(1-b_l)-\sum_{i\neq l}\tau_ib_i=0.
\end{align}
This and KKT conditions imply
\begin{align}
&\tau_l=\frac{t_s}{\mathbb{E}[T](1-b_l)}+\frac{\sum_{i\neq l}\tau_ib_i}{1-b_l},\label{eqa67}\\&
\tau_l\geq 0, r_l-b_l\left(\sum_{i=1}^Mr_i+1\right)\leq 0,\label{eqa68}\\&
\tau_l\left[r_l-b_l\left(\sum_{i=1}^Mr_i+1\right)\right]=0.\label{eqa69}
\end{align}
Since $\sum_{i=1}^Mb_i<1$, \eqref{eqa67} implies that $\tau_l>0$ for all $l$. This and \eqref{eqa69} result in 
\begin{equation}\label{eqa70}
r_l=b_l\left(\sum_{i=1}^Mr_i+1\right), \forall l.
\end{equation}
Because $\sum_{i=1}^Mb_i<1$, \eqref{eqa70} has a unique solution, which is given by 
\begin{equation}\label{eqa71}
r_l=\frac{b_l}{1-\sum_{i=1}^Mb_i}, \forall l.
\end{equation}
Hence, the solution to \eqref{lower_bound_case_bf_le_1_1} and \eqref{problem64} is given by \eqref{eqa71}. Substitute \eqref{eqa71} into \eqref{lower_bound_case_bf_le_1_1}, we get the following lower bound:
\begin{equation}\label{lower_bound_72}
\bar{\Delta}^{\text{w-peak}}_{\text{opt}}\ge \bar{\Delta}^{\text{w-peak}}_{\text{opt},2}\geq \sum_{l=1}^M\frac{w_le^{\frac{-\sum_{i=1}^Mb_i}{1-\sum_{i=1}^Mb_i}\frac{t_s}{\mathbb{E}[T]}}}{b_l}+\sum_{l=1}^Mw_l.
\end{equation}
This completes the proof. \qed

\section{Proof of Corollary \ref{corollary_centeralized_scheduler}}\label{Appendix_E}
We start by solving Problem \eqref{problem_centr} for optimal $\mathbf{a}$.  Problem \eqref{problem_centr} is a convex optimization problem and satisfies Slater's conditions. Thus, the KKT conditions are necessary and sufficient for optimality. Let $\mathbf{\lambda}=(\lambda_1,\ldots,\lambda_M)$ and $\nu$ be the Lagrange multipliers associated with the constraints \eqref{eq105} and \eqref{eq106}, respectively. Then, the Lagrangian of Problem \eqref{problem_centr} is given by
\begin{align}\label{lag119}
\begin{split}
L(\mathbf{a},\mathbf{\lambda},\nu)=&\sum_{i=1}^M\left[\frac{w_i}{a_i}+w_i\right]\\&+\sum_{i=1}^M\lambda_i(a_i-b_i)+\nu\left(\sum_{i=1}^Ma_i-1\right).
\end{split}
\end{align}
Take the derivative of \eqref{lag119} with respect to $a_l$ and set it equal to 0, we get
\begin{align}
\frac{-w_l}{a_l^2}+\lambda_l+\nu=0.
\end{align}
This and KKT conditions imply
\begin{align}
&a_l=\sqrt{\frac{w_l}{\lambda_l+\nu}},\\&
\lambda_l\ge 0,~a_l-b_l\le 0,\\&
\lambda_l(a_l-b_l)=0,\\&
\nu\ge 0,~\sum_{i=1}^Ma_i-1\le 0,\\&
\nu\left(\sum_{i=1}^Ma_i-1\right)=0.
\end{align}
If $\lambda_l=0$, then we have $a_l=\sqrt{w_l/\nu}$ and $a_l\le b_l$. This implies that $\nu>0$ and hence $\sum_{i=1}^Ma_i=1$, which holds when $\sum_{i=1}^Mb_i\ge 1$.

If $\lambda_l>0$, then we have $a_l=b_l$ and $a_l\le \sqrt{w_l/\nu}$. In this case, we either have $\nu>0$, which implies $\sum_{i=1}^Ma_i=1$ and this holds when $\sum_{i=1}^Mb_i\ge 1$; or $\nu=0$, which implies $\sum_{i=1}^Ma_i\leq 1$ and this holds when $\sum_{i=1}^Mb_i\leq 1$.

From the above argument, the solution can be driven according to the following two cases:

Case 1 (\textbf{Energy-adequate regime ($\sum_{i=1}^Mb_i\ge 1$)}): In this case, the optimal solution is given by
\begin{align}
a_l^{\star}=\min\left\lbrace b_l,\sqrt{\frac{w_l}{\nu^{\star}}}\right\rbrace,~\forall l,
\end{align}
where we must have $\nu^\star>0$, which implies $\sum_{i=1}^Ma^\star_i=1$. Hence, $\nu^\star$ satisfies
\begin{align}\label{eq126}
\sum_{i=1}^M\min\left\lbrace b_i,\sqrt{\frac{w_i}{\nu^{\star}}}\right\rbrace=1.
\end{align}
By comparing \eqref{eq126} with \eqref{condition_beta}, we can deduce that $\sqrt{1/\nu^{\star}}=\beta^{\star}$, where $\beta^{\star}$ satisfies
\begin{align}\label{eq111}
\sum_{i=1}^M\min\{b_i,\beta^{\star}\sqrt{w_i}\}=1.
\end{align}
Since $\sum_{i=1}^Mb_i\ge 1$, \eqref{eq111} has a solution for $\beta^{\star}$ as shown in Lemma \ref{lemma_solution_19}. Hence, the solution to Problem  \eqref{problem_centr} can be rewritten as
\begin{align}\label{sol127}
a^\star_l=\min\{b_l,\beta^{\star}\sqrt{w_l}\},~\forall l.
\end{align}
Substituting \eqref{sol127} into \eqref{problem_centr}, we obtain
\begin{align}
\bar{\Delta}_{\text{opt-s}}^{\text{w-peak}}=\sum_{i=1}^M\left[\frac{w_i}{\min\{b_i,\beta^{\star}\sqrt{w_i}\}}+w_i\right],
\end{align}
which is equal to the asymptotic optimal objective value of Problem \textbf{1} in energy-adequate regime in \eqref{asymptotic_value_final}. 

Case 2 (\textbf{Energy-scarce regime ($\sum_{i=1}^Mb_i< 1$)}): In this case, the optimal solution is
\begin{align}
a^\star_l=b_l,~\forall l.
\end{align}
 Substituting by this into \eqref{problem_centr}, we obtain
 \begin{align}
\bar{\Delta}_{\text{opt-s}}^{\text{w-peak}}=\sum_{i=1}^M\left[\frac{w_i}{b_i}+w_i\right],
\end{align}
which is equal to the asymptotic optimal objective  value of Problem \textbf{1} in energy-scarce regime in \eqref{asymptotic_value_final_2}. This completes the proof. \qed
\fi

%% file: sections/learning_proof.tex
\section{Proof of Theorem \ref{th:learn_regret}}\label{sec:learn_proof}
\subsection{Notation and Background on General State-Space Markov Processes}
While analyzing learning algorithm, we will have to work with Markov processes on general state-space~\cite{nummelin2004general,meyn2012markov}. In this section we provide a brief account of such processes.

\emph{Notation}: For a set of r.v. s $\cX$, we let $\cF(\cX)$ denote the smallest sigma-algebra with respect to which each r.v. in $\cX$ is measurable. For a set $\cX$, we let $\cX^{c}$ denote its complement. For an event $\cX$, we let $\id(\cX)$ denote its indicator random variable. For a set $\cX$, we let $\cB(\cX)$ denote the sigma-algebra of Borel sets of $\cX$.\par
We begin by showing that $\bs(n)$ can be taken to be the system state \slash sufficient statistics~\cite{striebel1965sufficient} in order to describe the sampled process. In what follows, we let $\cS := \bR_{+}\times \{0,1\}$. Denote by $\te :=  \int_{y=0}^{T_{\max}} y~f(y) dy$ the mean transmission time of a packet of any source, i.e., we use the abbreviation $\te = \bE\left[T\right]$. The proof of the following result is omitted for brevity.
\begin{lemma}\label{lemma:suff_stat}
Consider the system in which $M$ sources share a channel, and utilize the sleep period parameters as $\br(n)\equiv \br$ in order to modulate the sleep durations of sources. We then have that
\begin{align}\label{def:kernel_exist}
\bP\left(  \bs(n+1) \in A \big|  \cF_t \right) = \cK(\bs(n),\br,A;f),
\end{align}
where $\cF_t$ denotes the sigma-algebra generated by all the random variables until the $n$-th discrete sampling instant. 
The function $\cK$ is the kernel~\cite{nummelin2004general} associated with the controlled transition probabilities of the process $\bs(n)$, 
\begin{align}\label{def:kernel}
\cK: \cS^{M} \times \bR^{M}_+\times \cB(\bR^{M}) \mapsto [0,1].
\end{align}
Thus, $\cK(\bs,\br,A;f)$ is the probability with which the state at time $n+1$ belongs to the set $A$, given that the state at time $n$ is equal to $\bs$, and the vector comprising of sleep period parameters at time $n$ is equal to $\br$. Note that the kernel is parametrized by the density function of transmission time $f$.
\end{lemma}
We begin by stating some definitions associated with Markov Chains on \emph{General State-Spaces}. Though these can be found in standard textbooks on General State-Space Markov Chains such as~\cite{meyn2012markov,nummelin2004general}, we include them here in order to make the paper self-contained.  

Let us now fix the controls at $\br(n)\equiv \br$, and consider the resulting discrete-time Markov chain $\bs(n)\in \cS^{M}$. If $A$ is a Borel set, we let $P^{n}(\bx,A)$ denote the probability of the event $\bs(n)\in A$, given that $\bs(0) = \bx$.
\begin{definition}(Small Set)\label{def:small}
A set $C\in  \cB(\cS^{M})$ is called $\nu_m$ small if for all $\bx\in C$ we have that 
\begin{align*}
P^{m}(\bx,A) \ge \nu_m(A),~\forall  A\in \cB(\cS^M),
\end{align*}
for some non-trivial measure $\nu_m(\cdot)$ and some $m\in \bN$. 
\end{definition}

\begin{definition}(Petite Set)
Let $\mathbf{q} = \{q_n\}_{n\in \bN}$ be a probability distribution on $\bN$. A set $C\in  \cB(\cS^{M})$ and a non-trivial sub-probability measure $\nu_q(\cdot)$ are called petite if we have that
\begin{align*}
\sum_{n\in \bN} q_n P^{n}(\bx,A)\ge \nu_{q}(A), \forall A\in \cB(\cS^{M}),\forall x\in C. 
\end{align*}
\end{definition}

\begin{definition}(Strong Aperiodicity)
If there exists a $\nu_1$ small set $C$ such that we have $\nu_1(C)>0$, then the chain $\bs(n)$ is strongly aperiodic.  
\end{definition}
\subsection{Preliminary Results}\label{sec:jmls}
We now show that in order to minimize the expected value of $C(H)$, it suffices to design controllers that ``work directly'' with the sampled system. Thus, the quantity $\bs(n)$ as described in~\eqref{def:st} serves as a sufficient statistic for the purpose of optimizing the expectation of cumulative peak age~\cite{striebel1965sufficient}. We also show that this objective can be posed as a constrained Markov decision process~\cite{puterman2005markov}.
\begin{lemma}\label{lemma:equiv_problem}
Let $\bs(n), n=1,2,\ldots,$ be the sampled controlled Markov process. There exists a function $g: \cS^{M} \mapsto \bR$ so that $\bE\left[C(H)\right]$ in \eqref{def:peak_cost} is given by $\bE \left[\sum\limits_{n=1}^{H} g(\bs(n))\right]$.  
\end{lemma}
\begin{proof}
Consider the cumulative peak-age cost~\eqref{def:peak_cost} in which the $l$-th source incurs a penalty of $\Delta^{\text{peak}}_{l,i}$ upon delivery of the $i$-th packet. Let this delivery occur at the end of the $n$-th discrete time-slot (note that this time $n$ is random). Let us denote by $a^{\text{peak}}_{l}(n)$ the peak age of source $l$ during the (continuous) time interval (in the non-discretized system) corresponding to the discrete time slots $n-1$ and $n$. We could (instead of charging a penalty of $\Delta^{\text{peak}}_{l,i}$ units at the end of $n$-th slot) charge the quantity $\bE\left\{a^{peak}(n)| \bs(n-1),\br(n-1)\right\}$ at the discrete time instant $n-1$. For sources $k\neq l$ that are not transmitting between $n-1$ and $n$, and have $m_k(n-1)=0$, we let $g(\bs(n-1))=0$. It then follows from the law of the iterated expectations~\cite{billingsley2008probability} that the expected cost of the system under this modified cost function remains the same as that of the original system. This completes the proof. 
\end{proof}

\textbf{Ergodicity of $\bs(n)$}: We now derive a few useful results about the Markov process $\bs(n)$.
\begin{lemma}\label{lemma:petite}
Consider the multi-source wireless network operating under the controls $\br(n)\equiv \br$, and assume that the sensing time $t_s$ is sufficiently small, i.e., it satisfies $t_s<1$. Consider the associated process $\bs(n),~n=1,2,\ldots$ We then have the following:\\
\begin{enumerate}
\item Define
\begin{align*}
e_i := (M-i)\epsilon, \mbox{ and } m_i = 0, \forall i \in [N], 
\end{align*}
where $\epsilon>0$ is chosen to be sufficiently small. Consider the set 
\begin{align}\label{def:set_C}
C:= \otimes_{i=1}^{N} \left[ [ (M-i), (M-i) + e_i ] \times \left\{m_i\right\} \right].
\end{align}
The set $C$ is small for the process $\bs(n)$.
\item For the process $\bs(n)$, each compact set is petite.
\item The process $\bs(n)$ is strongly aperiodic.
\end{enumerate}
\end{lemma}
\begin{proof}
\par 
\begin{enumerate}
\item Consider $\bs(0)\in C$. It follows that at time $n=0$, all the sources are sleeping. Consider the following set denoted $C^{\prime}$: Sources $1$ and $2$ wake up within $t_s$ time duration of each other, while the other sources wake up much later than these two. Consequently, there is a collision between Source $1$ and Source $2$, and hence at time $n=1$ these two sources enter into sleep mode, so that at time $n=1$ all the sources are asleep. Also assume that the cumulative time elapsed for this event to occur is approximately equal to $t_s + \delta$, where $\delta>0$ is a sufficiently small parameter. The probability of the event $\left\{\bs(1)\in C^{\prime}\right\}$ can be lower bounded as follows 
\begin{align*}
\bP(\bs(1)\in C^{\prime} )\ge& \left( r_1  \int_{\delta}^{\delta+\epsilon} \exp(-r_1 x) dx\right)\\&\times t_s  r_2  \exp(-r_2 (\delta+\epsilon))\\&\times \left[\Pi_{i=3}^{N}  \int_{\delta+\epsilon}^{\infty} r_i \exp(-r_i x) dx\right].
\end{align*}
Since the above lower-bound on the probability of ``reaching $C^{\prime}$'' is true for all $\bs(0)\in C$, it follows from Definition~\ref{def:small} that the set $C$ is small.
\item Consider the process $\bs(n)$ starting in state $\bs(0)$, and let the age vector $a(0)$ belong to a compact set, so that $\bs(0)$ also belongs to a compact set.
We will derive a lower bound on the probability of the event $\left\{\bs(N)\in C\right\}$, where $C$ is as in~\eqref{def:set_C}. This will prove (ii) since we have already shown in (i) that the set $C$ is small. Consider the following sample path: at each time $i\in [1,M]$, we have that source $i$ successfully transmits a packet, and moreover the age of the packet received is approximately equal to $1$. We will derive a lower-bound on the probability of this event. In the following discussion we use $b>0$ and $\eta \in \left( 0, 1-t_s-b\right)$, where $\eta$ denotes the time when Source $1$ wakes up. Since the counter of the $i$-th source has a probability density equal to $r_i e^{-r_i x}$, the probability that during the $i$-th slot source $i$ gets channel access is lower bounded by $(1- \exp(-\eta r_i))\Pi_{j\neq i} e^{-r_j}  $; while the probability that the age of its delivered packet is around $1$, given that it wakes up at $\eta$, is lower bounded by $\int_{0}^{b} f(y)dy$. Thus, the probability of this sample path is lower bounded by
\begin{align*}
\Pi_{i=1}^{N} (1- \exp(-\eta r_i))\Pi_{j\neq i} e^{-r_j}   \int_{0}^{b} f(y)dy.
\end{align*}
This concludes the proof since along this sample path we have that $\bs(N)\in C$.
\item It follows from the discussion on page 121 of~\cite{meyn2012markov} that in order to prove the claim it suffices to show that the volume of the set $C\cap C^{\prime}$ is greater than $0$. However, this condition holds true if the parameter $\delta$ in (i) above has been chosen so as to satisfy $t_s + \delta< \epsilon$. 
\end{enumerate}
\end{proof}

We now show that the process $\bs(n)$ has a certain ``mixing property''.
For a measure $\mu$ and a function $f$, we define $\|\mu\|_f :=  \int f(x) d\mu(x) $.
\begin{lemma}[Geometric Ergodicity]\label{lemma:geometric}
Consider the controlled Markov process $\bs(n),~n=1,2,\ldots,$ associated with the network in which the controller utilizes $\br(n)\equiv \br $. The process $\bs(n)$ has an invariant probability measure, which we denote as $\pi(\infty,\br)$. Moreover, 
\begin{align}\label{ineq:geometric}
&\int \left(\|\by\|_1 +1\right)  d\left( P^{n}(\bx,\cdot) - \pi(\infty,\br)\right)(\by) \notag \\
&\le  R\left( \|\bs(0)\|_1 +1 \right) \rho^{n}, n\in \bN,
\end{align}
where $R>0$, and $\rho<1$.
\end{lemma}
\begin{proof}
Since we have shown in Lemma~\ref{lemma:petite} that $\bs(n)$ is strongly aperiodic, it follows from Theorem 6.3 of~\cite{meyn2012markov} that in order to prove the claim it suffices to show that the following holds true when $\|\bs(n+1)\|_1 $ is sufficiently large
\begin{align}\label{eq:cond_geometric}
\bE\left( \|\bs(n+1)\|_1 | \cF_{n} \right)  \le     \lambda \|\bs(n)\|_1 + L,
\end{align} 
where $\lambda<1$. Note that each source gets to transmit with a probability at least $\min_l \alpha_{l}$, and also the expected value of the inter-sampling time is upper-bounded by $\max\left\{\bE [T],\frac{\bE[T]}{\sum_{i=1}^Mr_i}+t_s\right\}$. It then follows that~\eqref{eq:cond_geometric} holds true with $\lambda$ set equal to $\min_l \alpha_{l}$, and $L$ equal to $\max\left\{\bE [T],\frac{\bE[T]}{\sum_{i=1}^Mr_i}+t_s\right\}$.
\end{proof}

\begin{lemma}(Differential Cost Function)\label{lemma:diff_co}
Consider the process $\bs(n),~n=1,2,\ldots,$ that describes the evolution of the network in which the controller utilizes $\br(n)\equiv \br$. Then, there exists a function $V: \cS^{M}\mapsto \bR $ that satisfies
\begin{align}\label{eq:diff_cost}
V(\bx) + \int  g(\bx) d\pi(\infty,\br) = g(\bx) + \int  \cK(\bx,\br,y;f) V(y)dy,
\end{align}
where $\cK$ is the transition kernel as described in Lemma~\ref{lemma:suff_stat}, the function $g$ is the one-step cost function as in Lemma~\eqref{lemma:equiv_problem}. Moreover, the function $V$ satisfies the following,
\begin{align}
V(\bx) \le \frac{R}{1-\rho} \left( \|\Delta(0)\|_1 +1 \right),
\end{align}
where the constant $R$ is as in Lemma~\ref{lemma:geometric}. 
\end{lemma}
\begin{proof}
We have shown in Lemma~\ref{lemma:geometric} that the process $\bs(n)$ is geometrically ergodic. Hence, it follows from Theorem 7.5.10 of~\cite{hernandez2012further} that there exists a function $V(\cdot)$ that satisfies~\eqref{eq:diff_cost}, and moreover it is given as follows, 
\begin{align*}
V(\bx) = \sum_{n=1}^{\infty} \left[\bE_{x}\left(g(\bx(n))\right) -  \int_{\cS^{M}}   g(\by)d \pi(\infty,\br)(\by)\right], x\in \cS.
\end{align*}
Substituting the geometric bound~\eqref{ineq:geometric} into the above, we obtain the following
\begin{align*}
V(\bx) &= \sum_{n=1}^{\infty} \left[\bE_{\bx}~ g(\bx(n)) -  \int_{\cS^{M}}   g(\by)d \pi(\infty,\br)(\by)\right]\\
&\le \sum_{n=1}^{\infty} \Big| \bE_{\bx}~ g(\bx(n)) -  \int_{\cS^{M}} g(\by)d \pi(\infty,\br)(\by)\Big|\\
&\le R\left( \|\bx(0)\|_1 +1 \right) \sum_{n=1}^{\infty} \rho^{n}\\
&= \frac{R\left( \|\bx(0)\|_1 +1 \right) }{1-\rho},
\end{align*}
where $\rho<1$. 
\end{proof}

\begin{lemma}[Smoothness properties of the optimal average cost]\label{lemma:convergence}
The optimal sleep period parameters $\br\ust_{\te}$ and average cost $\bar{\Delta}^{w-peak}$ satisfy the following:
\begin{enumerate}
\item We have that the function $\br\ust_{\te}: \Theta \mapsto \bR^{M}_{+}$ that maps the mean transmission time $\te$ to the optimal sleep period parameter, is a continuous function of $\theta$. Similarly, the average peak age is a continuous function of $\br$, i.e.,
\begin{align*}
\lim_{\br\to \br\ust_{\te}}\lim_{H\to\infty}\frac{1}{H} \sum_{n=1}^{H} \bE_{\br} \left[g(\bs(n))\right]  \\\to \lim_{H\to\infty}\frac{1}{H} \sum_{n=1}^{H} \bE_{\br\ust_{\te}} \left[g(\bs(n))\right], 
\end{align*}
where the sub-script $\br$ in the expectation $\bE_{\br}$ above refers to the fact that the averaging is performed w.r.t. the measure induced by the policy that uses sleep rates equal to $\br$.
\item The cumulative peak-age is locally Lipschitz continuous function of $\br$. Thus,
\begin{align*}
|\bar{\Delta}^{w-peak}(\br\ust_{\te};\te) - \bar{\Delta}^{w-peak}(\br;\te) | \le L_{1}  \|\br\ust_{\te} - \br\|,
\end{align*}
whenever $\|\br\ust_{\te} - \br \|$ is sufficiently small, and where the Lipschitz constant at  sleep period parameter $\br$ is given by 
\begin{align*}
L_{1} := \max_{i\in [M]} \frac{\partial \bar{\Delta}^{w-peak}}{\partial r_i}(\br).
\end{align*}
Similarly, the optimal sleep period parameter is a locally Lipschitz function of $\te$, so that we have,
\begin{align*}
\|\br\ust_{\te_1} - \br\ust_{\te_2} \|\le L_{2}  |\te_1 - \te_2|, L_2>0,
\end{align*}
whenever $ |\te_1 - \te_2|$ is sufficiently small. 
\end{enumerate}
In summary, there exists a $\delta>0$ such that whenever $|\te_1-\te_2|\le \delta$, then 
\begin{align*}
|\bar{\Delta}^{w-peak}(\br\ust_{\te_1};\te) - \bar{\Delta}^{w-peak}(\br\ust_{\te_2};\te) | \le L |\te_1 -\te_2|.
\end{align*}
\end{lemma}
\begin{proof}
Continuity of the functions under discussion is immediate from the relations \eqref{r_f_b_gr_1}, \eqref{x*}, \eqref{condition_beta}, \eqref{condition_x*_beta*}, \eqref{feasible_factor}, \eqref{feas_fact_eq27}. 
To prove the statement about Lipschitz continuity, it suffices to show that the average peak age is a Lipschitz continuous function of $\br$, and the optimal rate $\br\ust_{\te}$ is Lipschitz continuous function of $\te$. To prove this, it suffices to show that the average peak age is a continuously differentiable function of $\br$, and also $\br\ust_{\te}$ is a continuously differentiable function of $\te$ (see~\cite{hunter} for more details). The continuously differentiable property is evident from the relations~\eqref{t_avg_peak_age},~\eqref{r_f_b_gr_1}-\eqref{condition_beta} and~\eqref{condition_x*_beta*}-\eqref{feas_fact_eq27}. This completes the proof.
\end{proof}

\textbf{Bounds on the Estimation Error}: 
We now derive some concentration results for the estimate $\hat{\theta}(n)$ around the true value $\te\ust$. Let $\cC(n)$ be the confidence interval associated with the estimate $\hat{\theta}(n)$, i.e., 
\begin{align}\label{def:ucb_ci} 
\cC(n) := \left\{ \theta :  |\theta -\hat{\theta}(n) | \le  \xi(n), \theta >0 \right\},
\end{align} 
where
\begin{align*}
\xi(n) : =  T_{\max}\sqrt{\frac{2 \log\left(n^{\gamma}\right)}{N(n)}}, 1\le n \le H,
\end{align*}
$\gamma\ge 4$ is a constant, $N(n)$ is the total number of packet deliveries until $n$, and $T_{\max}$ is the maximum possible transmission time. We begin by showing that with a high probability, our confidence balls are true at all the times.
\begin{lemma}\label{lemma:g1}
Define 
\begin{align*}
\cG_1(n) := \left\{ \omega: \theta\ust \in \cC(n)\right\},
\end{align*}
where $\cC(n)$ is as in~\eqref{def:ucb_ci}, and $\te\ust$ is the vector consisting of true parameter values. We then have that
\begin{align*}
\bP\left( \cG^{c}_{1}(n)  \right) \le \frac{1}{n^{\gamma -1}}.
\end{align*}
\end{lemma}
\begin{proof}
Fix a positive integer $n_0$, and let $\hat{\te}$ denote the empirical estimate obtained from $n_0$ samples $T(1),T(2),\ldots,T(n_0)$ of the service times. It follows from Azuma-Hoeffding's inequality~\cite{hoeffding1994probability} that
\begin{align*}
\bP\left(  |\hat{\te} -  \te^{\star}  | >  x  \right)  \le \exp\left(-\frac{n_0 x^{2}}{2 T_{\max}^{2} }      \right).
\end{align*}
By using $x = T_{\max}\sqrt{ \frac{2\log \left( n^{\gamma}\right) }{n_0}    }$ in the above, we obtain,
\begin{align*}
\bP\left(  |\hat{\te} -  \te^{\star}  | >   T_{\max}\sqrt{ \frac{\log n^{\gamma} }{n_0}    }  \right) & \le \exp\left(-\log n^{\gamma}       \right)  \\
&=  \frac{1}{n^{\gamma}}.
\end{align*}
Since the total number of samples $n_0$ can assume values from the set $\left\{0,1,2,\ldots,n\right\}$, the proof then follows by using union bound on $n_0$.
\end{proof}
\begin{lemma}\label{lemma:g2}
Fix a $\delta_1 \in \left(0, p_{\min}\right)$, where $p_{\min}$ is as in~\eqref{def:cmin}. Define the event,
\begin{align}\label{def:g2}
\cG_2(n) := \left\{ \omega: N(n) > (p_{\min} -  \sqrt{\delta_1} ) n \right\}, 
\end{align}
where $N(n)$ denotes the number of samples that have been obtained until time $n$ for estimating transmission times. We then have that 
\begin{align*}
\bP(\cG^{c}_2(n) ) \le \exp\left( -\delta_1 n \right). 
\end{align*}
\end{lemma}
\begin{proof}
Consider the following martingale difference sequence $m(i)= \bE\left\{c(i)\big| \cF_{i-1}  \right\}  - c(n) $. Since $\bE\left\{c(i)\big| \cF_{i-1}  \right\}  \ge p_{\min} $, we have that
\begin{align}\label{ineq:1}
\sum_{i=1}^{n} m(i) \ge c_{\min} n - N(n).
\end{align}
Since $|m(i)|\le 1$, we have the following from Azuma-Hoeffding's inequality~\cite{hoeffding1994probability},
\begin{align*}
\bP\left(  \Big|\sum_{i=1}^{n} m(i)\Big|  \ge x  \right) \le \exp\left( -\frac{x^{2}}{n}  \right).
\end{align*}
Letting $x= \sqrt{\delta_1} n$, we get the following,
\begin{align}\label{ineq:2}
\bP\left(  \Big|\sum_{i=1}^{n} m(i)\Big|  \ge \sqrt{\delta_1} n  \right) \le \exp\left( -\delta_1 n  \right).
\end{align}
Substituting~\eqref{ineq:1} into the above inequality, we obtain 
\begin{align*}
\bP\left(  N(n)   \le \left( p_{\min} - \sqrt{\delta_1} \right) n  \right) \le \exp\left( -\delta_1 n  \right).
\end{align*}
This completes the proof. 
\end{proof}

\subsection{Regret Analysis}
\vspace{.5cm}

The cumulative regret $R(H)$~\eqref{def:regret} decomposes into the sum of ``episodic regrets'' $R^{(e)}(k)$ as follows:
\begin{align}
\bE\left[R(H)\right] &= \sum_{k=1}^{K} \bE\left[~R^{(e)}(k) \right],\label{reg:dec_1}\\
\mbox{ where } R^{(e)}(k) :&= \bE\left\{\sum_{n\in \cE_k}g(\bs(n)) - \bar{\Delta}^{\text{w-peak}}(\mathbf{r}^{\star}) \Big|  \cF_{\tau_k} \right\}. \label{reg:dec_2}
\end{align}
Combining the regret decomposition with the smoothness properties of the optimal average cost that were derived in Lemma~\ref{lemma:convergence}, we obtain the following key result that allows us to upper-bound $R(H)$. 
\begin{lemma}\label{lemma:regret_decompose}
The cumulative expected regret~\eqref{reg:dec_1} for a learning algorithm can be upper-bounded as follows,
\begin{align}
&\bE \left[R(H) \right]\le K_2 \sum_{k=1}^{K} \left(\tau_{k+1} - \tau_k  \right) \bP( | \hat{\te}(\tau_k) - \te\ust |>\delta ) \notag\\
&+L \sum_{k=1}^{N} (\tau_{k+1} - \tau_k) \bE \left(  | \hat{\te}(\tau_k) - \te\ust |\id\left\{| \hat{\te}(\tau_k) - \te\ust |\le \delta \right\} \right),\label{def:reg_dec}
\end{align}
where the constant $\delta>0$ is as in Lemma~\ref{lemma:convergence}.
\end{lemma}
\begin{proof}
It follows from the ergodicity properties of the process $\bs(n)$ that were proved in Lemma~\ref{lemma:diff_co} and Assumption~\ref{assum:bounded} regarding $\bs(n)$, that the episodic regret can be bounded as follows ($\rho,R$ are as in Lemma~\ref{lemma:diff_co} and Assumption~\ref{assum:bounded}), 
\begin{align}\label{bound:re}
R^{(e)}(k) \le & \frac{R}{1-\rho} \left( K_1 +1 \right)\\&+  \Big|\bar{\Delta}^{w-peak}(\br\ust_{\hat{\te}(\tau_k)};\te) - \bar{\Delta}^{\text{w-peak}}(\mathbf{r}^{\star}) \Big|\left(\tau_{k+1} - \tau_k  \right) .
\end{align}
The following two events are possible:
\begin{enumerate}[(i)]
\item $|\te\ust - \hat{\te}(\tau_k)|< \delta$: In this case it follows from Lemma~\ref{lemma:convergence} that
\begin{align*}
\Big|\bar{\Delta}^{w-peak}(\br\ust_{\hat{\te}(\tau_k)};\te) - \bar{\Delta}^{\text{w-peak}}(\mathbf{r}^{\star}) \Big| \le L |\te\ust - \hat{\te}(\tau_k)|.
\end{align*}
\item  $|\te\ust - \hat{\te}(\tau_k)|> \delta$: It follows from Assumption~\ref{assum:2} that the average performance under any sleep parameter cannot exceed $K_2$, and hence we can bound $\Big|\bar{\Delta}^{w-peak}(\br\ust_{\hat{\te}(\tau_k)};\te) - \bar{\Delta}^{\text{w-peak}}(\mathbf{r}^{\star}) \Big| $ by $K_2$.
\end{enumerate}
The proof then follows by substituting the bounds discussed above for the two cases into~\eqref{bound:re}, and using regret decomposition result.
\end{proof}
We now separately bound the expressions obtained in the two events ($|\te\ust - \hat{\te}(\tau_k)|< \delta$, $|\te\ust - \hat{\te}(\tau_k)|>\delta$).
\begin{flushleft}
\textbf{Regret when} $\bm{|\te\ust - \hat{\te}(\tau_k)|> \delta}$: 
\end{flushleft}

Choose a sufficiently large $k_0\in \bN$ that satisfies 
\begin{align}\label{def:k0}
\tau_{k_0} = O\left(\frac{1}{\delta_1}\log H\right).
\end{align}
Define the following event
\begin{align*}
\cG_3 :=  \cap_{k \ge k_0 } \cG_2(\tau_k).
\end{align*}
By combining the result of Lemma~\ref{lemma:g2} with the union bound and using~\eqref{def:k0} we conclude that $\cG_3$ has a probability greater than $1- \sum_{k>k_0} \exp(-\delta_1 \tau_k) = 1-O\left(\frac{1}{H}\right)$. On $\cG_3$, the number of samples $N(\tau_k)$ at the beginning of each episode $k>k_0$ is greater than $\left(p_{\min} -  \sqrt{\delta_1}\right)\tau_k$. Thus on $\cG_3$, for episodes $k>k_0$ the radius of $\cC(\tau_k)$ is less than $\sqrt{\frac{\gamma\log H}{(p_{\min}- \sqrt{\delta_1})\tau_k }}$. Let $k_1$ be the smallest integer that satisfies
\begin{align}\label{ineq:radius}
\frac{\gamma\log H}{( p_{\min} - \sqrt{\delta_1})\tau_{k_{1} }}\le \delta^{2}, \mbox{ i.e. } \tau_{k_1} \ge \frac{1}{( p_{\min} - \sqrt{\delta_1})\delta^2 } \gamma\log H,
\end{align} 
where the constant $\delta>0$ is as in Lemma~\ref{lemma:convergence}. Thus on $\cG_3$, for episodes $k\ge \max \left\{ k_0, k_1 \right\}$, the radius of confidence intervals is less than $\delta$. Note that on $\cap_{k}\cG_1(\tau_k)$ the confidence intervals~\eqref{def:ucb_ci} at the beginning of each episode are true. Hence, on $\left\{\cap_{k}\cG_1(\tau_k) \right\}\cap \cG_3$ we have $| \hat{\te}(\tau_k) - \te\ust |<\delta$ for epsiodes $k\ge \max \left\{ k_0, k_1 \right\}$. Thus, on $\left\{\cap_{k}\cG_1(\tau_k) \right\}\cap \cG_3$ this regret is bounded by $K_2 \max\left\{\tau_{k_0},\tau_{k_1}\right\}$. Now consider sample paths for which some of the confidence intervals fail. The probability that $\cC(\tau_k)$ fails is less than $\frac{1}{\tau^{\gamma -1}_k}$ (Lemma~\ref{lemma:g1}); moreover since the episode duration of $\cE_k$, $(\tau_{k+1} - \tau_k )$ is less than $\tau_k$, we have that the expected value of the regret during $\cE_k$ in the event of failure of $\cC(\tau_k)$ is less than $K_2 \frac{1}{\tau^{\gamma - 2}_k}$. Since $\gamma\ge 4$, the cumulative expected regret arising from this is bounded by $K_2 \sum_k \frac{1}{\tau^{\gamma - 2}_k}\le K_2 \frac{\pi^2}{6}$~\cite{ayoub1974euler}. We summarize our discussion as follows. 
\begin{lemma}\label{lemma:first_t_b}
Under Algorithm~\ref{algo:ucb} the following is true, 
\begin{align}
&\sum_{k=1}^{K} \left(\tau_{k+1} - \tau_k  \right) \bP( | \hat{\te}(\tau_k) - \te\ust |>\delta ) \notag\\& \le K_2 \max\left\{\frac{\gamma\log H}{( p_{\min} - \sqrt{\delta_1})  \delta^{2}},O\left(\frac{1}{\delta_1}\log H\right)  \right\}+ K_2 \frac{\pi^2}{6},
\end{align}
where $\gamma \ge 4$.
\end{lemma}
 
\begin{flushleft}
\textbf{Regret when} $\bm{|\te\ust - \hat{\te}(\tau_k)| < \delta}$: 
\end{flushleft}
As discussed above, on $\cap_{k}\cG_1(\tau_k) \cap \cG_3$ we have  $|\te(\tau_k) -\te\ust|<\delta$ for episodes $k>k_1$. Thus, after using the smoothness property of optimal average cost that was developed in Lemma~\ref{lemma:convergence}, we obtain that the second summation in the r.h.s. of~\eqref{def:reg_dec} can be bounded by the following quantity,
\begin{align*}
\sum_{k>k_1} \left(\tau_{k+1}-\tau_k\right) \sqrt{\frac{\gamma\log H}{( p_{\min} - \sqrt{\delta_1})\tau_k }}.
\end{align*}
Since we have $\tau_{k+1} - \tau_k \le \tau_k$, the above can be bounded by $\sqrt{\frac{\gamma\log H}{ ( p_{\min} - \sqrt{\delta_1}) }  } \sum_{k>k_1}  \sqrt{\tau_k}$. By using Cauchy Schwart'z inequality, the quantity $\sum_{k>k_1}  \sqrt{\tau_k}$ can be upper-bounded as $\sqrt{H K}$, where $K$ denotes the number of episodes. Since $K= O\left(\log H\right)$, this regret is bounded by $\sqrt{\frac{H\gamma (\log H)^2}{ ( p_{\min} - \sqrt{\delta_1}) }  }$. The bound we discussed is summarized below.
\begin{lemma}\label{lemma:second_t_b}
Under Algorithm~\ref{algo:ucb} the following is true, 
\begin{align}
&L \sum_{k=1}^{N} (\tau_{k+1} - \tau_k) \bE \left(  | \hat{\te}(\tau_k) - \te\ust |\id\left\{| \hat{\te}(\tau_k) - \te\ust |\le \delta \right\} \right)\notag \\&\le L\sqrt{\frac{H\gamma (\log H)^2}{ ( p_{\min} - \sqrt{\delta_1}) }  }.
\end{align}
\end{lemma}
We are now in a position to prove main result Theorem~\ref{th:learn_regret}.
\begin{proof}(Theorem~\ref{th:learn_regret})
The proof follows by substituting the bounds obtained in Lemma~\ref{lemma:first_t_b} and Lemma~\ref{lemma:second_t_b} into the regret decomposition result of Lemma~\ref{lemma:regret_decompose}.
\end{proof}